\documentclass[english]{article}
\usepackage[T1]{fontenc}
\usepackage[latin9]{inputenc}
\usepackage{geometry}
\usepackage{amsmath}
\usepackage{amsthm}
\usepackage{amssymb}
\usepackage{graphicx}
\usepackage{esint}
\usepackage[nocompress]{cite}

\usepackage{color}
\usepackage{soul}

\usepackage{authblk}

\allowdisplaybreaks

\makeatletter
\theoremstyle{plain}
\newtheorem{thm}{\protect\theoremname}[section]
\theoremstyle{plain}
\newtheorem{lem}[thm]{\protect\lemmaname}
\theoremstyle{plain}
\newtheorem{prop}[thm]{\protect\propositionname}
\theoremstyle{remark}
\newtheorem{rem}[thm]{\protect\remarkname}
\theoremstyle{remark}
\newtheorem*{rem*}{\protect\remarkname}

\DeclareMathOperator{\Tr}{Tr}

\makeatother

\usepackage{babel}
\providecommand{\lemmaname}{Lemma}
\providecommand{\propositionname}{Proposition}
\providecommand{\remarkname}{Remark}
\providecommand{\theoremname}{Theorem}

\numberwithin{equation}{section} 
\numberwithin{thm}{section}

\begin{document}
	
	\title{Rate of Convergence towards Hartree Dynamics with Singular Interaction Potential}
	\author[1]{Li Chen\thanks{chen@math.uni-mannheim.de}}
	\author[2]{Ji Oon Lee\thanks{jioon.lee@kaist.edu}}
	\author[2]{Jinyeop Lee\thanks{jinyeoplee@kaist.ac.kr}}
	\affil[1]{Institut f\"ur Mathematik, Universit\"at Mannheim}
	\affil[2]{Department of Mathematical Sciences,  KAIST}
	\date{}
	
	\maketitle
	
	\begin{abstract}
		We consider a system of $N$-Bosons with a two-body interaction potential $V \in L^2(\mathbb{R}^3)+L^\infty(\mathbb{R}^3)$, possibly {more} singular than the Coulomb interaction. We show that, with $H^1(\mathbb{R}^3)$ initial data, the difference between the many-body Schr\"{o}dinger evolution in the mean-field regime and the corresponding Hartree dynamics is of order $1/N$, for any fixed time. The $N$-dependence of the bound is optimal.
	\end{abstract}
	
	\section{Introduction} \label{sec:intro}
	
	We consider the dynamics of $N$-Bosons in three dimensions interacting through a symmetric two-body potential. The system is described on the Hilbert space $L_{s}^{2}\left(\mathbb{R}^{3N}\right)$, the
	subspace of $L^{2}(\mathbb{R}^{3N})$, consisting of all symmetric functions with respect to the permutation of particles. Its dynamics is governed by a mean-field Hamiltonian of the form
	\begin{equation} \label{eq:N_body_Hamiltonian}
	H_{N}=\sum_{j=1}^{N}\left(-\Delta_{x_{j}}\right)+\frac{1}{N}\sum_{i<j}^{N}V\left(x_{i}-x_{j}\right)
	\end{equation}
	with the interaction potential $V$, which will be specified later. Note that the coupling constant $1/N$ in the interaction term ensures that the kinetic energy and the interaction potential energy are comparable in terms of $N$, hence the two energies compete with each other and generate the nontrivial effective equation for the macroscopic dynamics of the system.
	
	Suppose that the system is fully condensated, i.e., the wave function is given by
	\begin{equation}
	\psi_{N}\left(\mathbf{x}\right)=\prod_{j=1}^{N}\varphi\left(x_{j}\right) \qquad\text{for some }\varphi\in H^{1}\left(\mathbb{R}^{3}\right)\label{eq:Initial_Wave}
	\end{equation}
	with $\mathbf{x}=\left(x_{1},x_{2},{{\dots}},x_{N}\right)\in\mathbb{R}^{3N}$ and normalization $\left\Vert \varphi\right\Vert _{L^{2}\left(\mathbb{R}^{3}\right)}=1$. The time evolution $\psi_{N,t}$ of \eqref{eq:Initial_Wave} is described by time-dependent Schr\"odinger equation
	\begin{equation}
	\mathrm{i}\partial_{t}\psi_{N,t}=H_{N}\psi_{N,t}\quad\text{with }\psi_{N,0}=\psi_{N}.\label{eq:N_body_Sch_eq}
	\end{equation}
	The solution of \eqref{eq:N_body_Sch_eq} can also be written as $\psi_{N,t}=e^{-\mathrm{i} H_{N}t}\psi_{N}$.
	
	In the case of free evolution where $V=0$, it can be easily checked that $\psi_{N,t} = \prod_{j=1}^{N}\varphi_{t}\left(x_{j}\right)$ with $\varphi_{t} = e^{-i \Delta t} \varphi$. However, factorization cannot be preserved under the presence of the interaction, and we can only expect that the wave function is approximately factorized. Heuristically, if we assume the approximate factorization
	\begin{equation}
	\psi_{N,t}\simeq\prod_{j=1}^{N}\varphi_{t}\left(x_{j}\right)\quad\text{for large }N, \label{eq:Factorization_of_Wave}
	\end{equation}
	then the total potential experienced by the particle $x_1$ can be approximated by
	\begin{equation}
	\frac{1}{N} \sum_{j=2}^{N} \int_{\mathbb{R}^3} V(x_j - x_1) |\varphi_{t} (x_j)|^2 \mathrm{d} x_j \simeq (V*\left|\varphi_{t}\right|^{2})\left(x_1 \right).
	\end{equation}
	Thus, the evolution of the one-particle wave function $\varphi_{t}$ can be described approximately by the nonlinear Hartree equation
	\begin{equation}
	\mathrm{i}\partial_{t}\varphi_{t}=-\Delta\varphi_{t}+(V*\left|\varphi_{t}\right|^{2})\varphi_{t}\label{eq:Hartree_eq}
	\end{equation}
	{with initial data $\varphi_{t=0}=\varphi_0$.}
	
	To understand \eqref{eq:Factorization_of_Wave} mathematically, we first need to give a precise meaning to the approximation in it. While it seems to be natural to consider the convergence in $L^2$-norm, it turns out that the norm convergence does not hold directly and second order correction must be introduced to obtain a norm approximation \cite{Grillakis2010,Grillakis2011}. We instead compare the {two} sides of \eqref{eq:Factorization_of_Wave} by their marginal densities, or reduced density matrices. The density matrix $\gamma_{N,t}=\left|\psi_{N,t}\right\rangle \left\langle \psi_{N,t}\right|$ associated with $\psi_{N,t}$ is defined as the orthogonal projection onto $\psi_{N,t}$. The kernel of $\gamma_{N,t}$ is thus given by
	\[
	\gamma_{N,t}\left(\mathbf{x};\mathbf{x}'\right)=\psi_{N,t}\left(\mathbf{x}\right)\overline{\psi_{N,t}}\left(\mathbf{x}'\right).
	\]
	We define the $k$-particle marginal density through its kernel
	\begin{equation}
	\gamma_{N,t}^{\left(k\right)}\left(\mathbf{x}_{k};\mathbf{x}'_{k}\right)=\int \mathrm{d} \mathbf{x}_{N-k}\gamma_{N,t}\left(\mathbf{x}_{k},\mathbf{x}_{N-k};\mathbf{x}'_{k},\mathbf{x}_{N-k}\right).\label{eq:Kernel_of_Marginal_Density}
	\end{equation}
	Since $\left\Vert \psi_{N,t}\right\Vert _{L^{2}\left(\mathbb{R}^{3N}\right)}=1$, we can see {that} $\Tr\gamma_{N,t}^{\left(k\right)}=1$ for $1\leq k\leq N$ and for any $t\in\mathbb{R}$. In particular, {$\gamma^{(k)}_{N,t}$} is a trace class operator.
	
	For a large class of {interactions} $V$, in the large $N$ limit, the $k$-particle marginal density associated with $\psi_{N,t}$ converges to the $k$-particle marginal density associated with the factorized wave function $\varphi^{\otimes N}$, under the condition that $\varphi\in H^{1}(\mathbb{R}^{3})$. More precisely, for any fixed $t\in\mathbb{R}$,
	\begin{equation}
	\operatorname{Tr}\left|\gamma_{N,t}^{\left(k\right)}-\left|\varphi_{t}\right\rangle \left\langle \varphi_{t}\right|^{\otimes k}\right|\to 0\quad\text{as }N\to\infty.\label{eq:Tr_convergence}
	\end{equation}
	where $\varphi_{t}$ is a solution of the non-linear Hartree equation \eqref{eq:Hartree_eq}.
	The first rigorous proof of \eqref{eq:Tr_convergence} was obtained by Spohn \cite{Spohn1980} for a bounded interaction $V$ and it was later extended by Erd\H{o}s and Yau \cite{Erdos2001} to the Coulomb type interaction. The main technique in these proofs was the analysis of BBGKY hierarchy.
	
	A natural question arising from the study of \eqref{eq:Tr_convergence} is the rate of the convergence. The first explicit bound on the rate of convergence,
	\begin{equation}
	\operatorname{Tr}\left|\gamma_{N,t}^{\left(k\right)}-\left|\varphi_{t}\right\rangle \left\langle \varphi_{t}\right|^{\otimes k}\right| \leq \frac{C e^{Kt}}{\sqrt N}, \label{eq:rate_sqrt_N}
	\end{equation}
	was obtained by Rodnianski and Schlein \cite{Rodnianski2009}, where the constants $C$ and $K$ do not depend on $N$ and $t$. The proof is based on the coherent state approach introduced by Hepp \cite{Hepp1974} and extended by Ginibre and Velo \cite{Ginibre1979_1,Ginibre1979_2}. The $N$-dependence in \eqref{eq:rate_sqrt_N} is not optimal, and an optimal bound of $O(N^{-1})$ was obtained by Erd\H{o}s and Schlein \cite{erdHos2009quantum} for bounded interactions, using a method inspired by Lieb-Robinson bound. The optimal bound for Coulomb type interaction was proved in \cite{Chen2011a}, which extended the technique provided in \cite{Rodnianski2009}.
	
	The results on the convergence and its rate can further be extended to a system with a pair interaction {{more}} singular than the Coulomb potential. Pickl \cite{Pickl2011simple} introduced a new strategy using a functional that counts the relative number of particles not in the state $\varphi_t$. Based on \cite{Pickl2011simple}, Knowles and Pickl proved in \cite{knowles2010mean} that \eqref{eq:Tr_convergence} holds with the rate of convergence $O(N^{-1/2})$ for $V \in L^p (\mathbb{R}^3) + L^{\infty} (\mathbb{R}^3)$ with $p \geq 2$ (and the rate that deteriorates with $p$ for $p \leq 6/5$ with stronger assumption on the initial condition $\varphi$). The method in \cite{knowles2010mean}, which utilizes the projectors, is also applicable to the semi-relativistic case where the Laplacian in \eqref{eq:N_body_Hamiltonian} is replaced by $\sqrt{1-\Delta_{x_j}}$. For the convergence results and the optimal rate of convergence in the semi-relativistic case, we refer to \cite{Elgart2007,Michelangeli2012,Lee2013}.
	We remark that the rate of convergence towards cubic Nonlinear Schr\"odinger equation in the Gross-Pitaevskii regime was considered in \cite{Pickl2010,NielsOliveiraSchlein}, and recently, the rate of convergence $1/N$ is obtained in \cite{2017arXiv170205625B} in a different sense from the one used in this paper.
	
	In this paper, we prove the optimal rate of convergence in \eqref{eq:Tr_convergence} with the interaction more singular than the Coulomb case. The main result of this work is the following theorem:
	
	\begin{thm} \label{thm:Main_Theorem}
		Suppose that $V=V_2+V_\infty$ with $V_2\in L^2(\mathbb{R}^3)$ and $V_\infty\in L^\infty (\mathbb{R}^3)$. Let $\varphi_{t}$ be the solution of the Hartree equation \eqref{eq:Hartree_eq} with initial data $\varphi_{0}=\varphi\in H^{1}(\mathbb{R}^{3})$. Let $\psi_{N,t}=e^{-\mathrm{i}H_{N}t}\varphi^{\otimes N}$ and $\gamma_{N,t}^{\left(1\right)}$ be the one particle reduced density associated with $\psi_{N,t}$ as defined in \eqref{eq:Kernel_of_Marginal_Density}. Then there exist constants $C$ and $K$, depending only on $\left\Vert \varphi\right\Vert _{H^{1}\left(\mathbb{R}^{3}\right)}$, $\|V_2\|_{L^2(\mathbb{R}^3)}$, and $\|V_\infty\|_{L^\infty(\mathbb{R}^3)}$, such that
		\begin{equation}
		\operatorname{Tr}\left|\gamma_{N,t}^{(1)}-\left|\varphi_{t}\right\rangle \left\langle \varphi_{t}\right|\right|\leq\frac{C e^{K{t^{3/2}}}}{N}.
		\end{equation}
	\end{thm}
	We remark that the assumption $\varphi \in H^1(\mathbb{R}^3)$ in Theorem \ref{thm:Main_Theorem} is natural in the sense that both the mass and the kinetic energy of each particle are bounded. With the assumption $\varphi \in H^1(\mathbb{R}^3)$, Theorem \ref{thm:Main_Theorem} covers the case {{of}} singular initial wave function, i.e., $\|\varphi\|_{L^\infty} = \infty$.

	To prove Theorem \ref{thm:Main_Theorem}, as in \cite{Rodnianski2009, Chen2011a}, we first prove the optimal rate of convergence for coherent states and convert it to the fixed-number particle case. As we will see in Section \ref{sec:prelim}, the marginal density is closely related to the second quantization operators in the Fock space representation, and the fluctuation of the second quantization operators acting on the coherent states can be identified with a unitary operator $\mathcal{U}(t;s)$ (see \eqref{eq:def_mathcalU}). 
	{
		In \cite{Chen2011a}, to obtain an estimate on $\mathcal{U}(t;s)$, another unitary operator $\mathcal{U}_2 (t;s)$ with a simpler generator was introduced; the difference between $\mathcal{U}(t;s)$ and $\mathcal{U}_2 (t;s)$ was controlled by the kinetic energy term along the dynamics generated by $\mathcal{U}_2 (t;s)$ (see (4.10) and (4.11) of \cite{Chen2011a}).}(
	In \cite{Chen2011a}, to obtain an estimate on $\mathcal{U}(t;s)$, another unitary operator $\mathcal{U}_2 (t;s)$ with a simpler generator was introduced and the difference between two unitary operator was controlled by the kinetic energy term with respect to $\mathcal{U}_2 (t;s)$ (see (4.10) and (4.11) of \cite{Chen2011a}).)
	
	The first obstacle we face with a more singular interaction is the control of the kinetic energy term. The approach in \cite{Chen2011a} is not directly applicable to this case, since many inequalities used in the estimate of the kinetic energy, most notably the Hardy inequality $V^2\leq D(1-\Delta)$, {are} no longer applicable. Without the control of the kinetic energy, we do not have an upper bound for the generator {$\mathcal{L}_2$} in \eqref{eq:L_2}, hence the comparison between $\mathcal{U}(t;s)$ {and} $\mathcal{U}_2 (t;s)$ is not possible. To overcome the difficulty, we introduce another unitary operator $\tilde{\mathcal{U}}(t;s)$ that contains $\mathcal{L}_4$ in its generator (see \eqref{eq:def_mathcaltildeU} and \eqref{eq:L_4}). However, this brings yet another obstacle, since the  proof of the optimal bound in \cite{Chen2011a} uses {the} special property of $\mathcal{U}_2 (t;s)$ that a certain evolution of the vacuum state under it lives exclusively in the one-particle sector of the Fock space (see Lemma 8.1 of \cite{Chen2011a}). We thus adopt the idea in \cite{Lee2013} that, although the evolution of the vacuum state under $\tilde{\mathcal{U}}(t;s)$ can reside in any sector with an odd number of particles, we can gain an additional factor of $O(N^{-1/2})$ as long as the number of particles is odd (see Lemma \ref{lem:coherent_even_odd}).
	
	Even without the use of the kinetic energy term, we need better regularity of $\varphi_t$ to compensate the singularity of $V$ in various estimates. To gain the regularity, we follow \cite{Chen2011} and apply the Strichartz's estimate of Hartree evolution. The Strichartz's estimate appears in many different places in the paper, most notably in the proof of an upper bound on the fluctuation of the expected number of particles under the evolution $\mathcal{U}(t;s)$ (see Lemma \ref{lem:NjU}).
	
	The paper is organized as follows. In Section \ref{sec:prelim}, we introduce some notions and preliminary results that will be used in the proof of the main theorem, including the Stricartz's estimate of the Hartree evolution and the Fock space representation. In Section \ref{sec:Pf-of-Main-Thm}, we prove {our main result}, Theorem \ref{thm:Main_Theorem}. The proof of propositions and lemmas used in the proof of the theorem is provided in Section \ref{sec:lemmas}.

	\begin{rem}
		The assumption in Theorem \ref{thm:Main_Theorem} is satisfied by any interaction potential $V$ of the form $V(x) = \sum_{i=1}^{n}\lambda_i |x|^{-\gamma_i}+c$ with a positive integer $n$, an exponent $0<\gamma_i<3/2$, interaction strength $\lambda_i\in\mathbb{R}$, and offset $c\in\mathbb{R}$.
	\end{rem}

	\begin{rem}[Notational Remark] \mbox{ }
		\begin{enumerate}
			\item The subindex $t$ in $\varphi_t$ and other functions emphasizes the time-dependence of the functions; it does not mean the partial derivative of the function with respect to time $t$.
			\item For a function $f:(\mathbb{R},\mathbb{R}^3)\to \mathbb{C}$, we denote $f_t(x):=f(t,x)$ and let $\|f_t\|$ be the $L^2(\mathbb{R}^3,\mathrm{d}x)$-norm with respect to spatial variable for notational simplicity. Other $L^p(\mathbb{R}^3)$-norms are denoted by $\|f_t\|_p$ or $\|f_t\|_{L^p}$. Similarly, $\|f_t\|_{H^s}$ is the $H^s(\mathbb{R}^3)$-norm.
			Furthermore, for any measurable set $S\subset \mathbb {R}$ we use the {following} notation:
			\[
			\|f\|_{L^p(S,L^q(\mathbb{R}^3))}:=\left( \int_S \mathrm{d}s\,\left( \int_{\mathbb{R}^3} \mathrm{d}x\, |f(s,x)|^q\right)^{\frac{p}{q}}\right)^{\frac{1}{p}}
			\]
			\item Constants $C$ and $K$ may differ line by line and we do not track them precisely as long as they are independent of the number of particles $N$ and time $t$.
		\end{enumerate}
		
	\end{rem}

	\section{Preliminary} \label{sec:prelim}
	
	In this section, we introduce basic notions and results that will be used throughout the rest of the paper.
	
	\subsection{Strichartz's estimate of Hartree evolution} \label{sec:decaying_Hartree}
	
	This subsection is devoted to explain Strichartz's estimate of Hartree evolution by giving the global existence theorem for Cauchy problem \eqref{eq:Hartree_eq} with initial data $\varphi_{0}$, i.e., if $\varphi_0\in H^1(\mathbb{R}^3)$, then $\varphi_t\in C((0,t),H^1(\mathbb{R}^3)) \cap L^2((0,t),L^\infty(\mathbb{R}^3))$. The following lemma shows that $\|\varphi_t\|_{H^1}$ is bounded uniformly on $t$.
	
	\begin{lem} \label{lem:H^1 bound}
		For the solution $\varphi_t$ of Hartree equation \eqref{eq:Hartree_eq} with $V\in L^2(\mathbb{R}^3)+L^\infty(\mathbb{R}^3)$ and $\varphi_0\in H^1(\mathbb{R}^{3})$, { there exists a constant $C$, depending only on $\left\Vert \varphi\right\Vert _{H^{1}}$, $\|V_2\|_{L^2}$, and $\|V_\infty\|_{L^\infty}$, such that}
		\[
		\|\varphi_{t}\|_{H^{1}}\leq C,
		\]
	\end{lem}
	
	\begin{proof}
		Let $V=V_2 + V_\infty$ so that $V_2\in L^2(\mathbb{R}^3)$ and $V_\infty \in L^\infty(\mathbb{R}^3)$. Define the energy $\mathcal{E}(\varphi_t)$ by	
		\begin{align}
		\mathcal{E}(\varphi_t)&:=\frac{1}{2}\int \mathrm{d}x\,|\nabla\varphi_t(x)|^2 + \frac{1}{4} \int\mathrm{d}x\mathrm{d}y\,V(x-y)|\varphi_t(x)|^2|\varphi_t(y)|^2\\
		&= \frac{1}{2}\int \mathrm{d}x\, |\nabla\varphi_t(x)|^2 \notag \\
		&\qquad + \frac{1}{4} \int \mathrm{d}x \mathrm{d}y \, V_2(x-y)|\varphi_t(x)|^2|\varphi_t(y)|^2 + \frac{1}{4} \int \mathrm{d}x \mathrm{d}y \, V_{\infty} (x-y)|\varphi_t(x)|^2|\varphi_t(y)|^2. \notag
		\end{align}
		To estimate the first term {of last line}, we use Young's inequality, Riesz-Thorin interpolation theorem, and Sobolev inequality to find that
		\begin{equation}
		\begin{split}
		&\int\mathrm{d}x\mathrm{d}y\,V_{2}(x-y)|\varphi_{t}(x)|^{2}|\varphi_{t}(y)|^{2} \leq C \|V_2\|_{L^2}\||\varphi_t|^2\|_{L^2}\|\varphi_t\|_{L^2}^2\leq C\|V_2\|_{L^2}\|\varphi_t\|_{L^4}^2\|\varphi_t\|_{L^2}^2\\
		&\quad\leq C\|V_2\|_{L^2}\|\varphi_t\|_{L^2}^{5/2}\|\varphi_t\|_{L^6}^{3/2}\leq C\|V_2\|_{L^2}\|\varphi_t\|_{L^2}^{5/2}\|\varphi_t\|_{H^1}^{3/2} \leq C\|V_2\|_{L^2}\|\varphi_t\|_{L^2}^{5/2}(\varepsilon \|\varphi_t\|_{H^1}^2 + \varepsilon^{-3}).
		\end{split}
		\end{equation}
		Thus, from the mass conservation $\| \varphi_t \|_{L^2} = 1$ we find that
		\begin{align}
		\mathcal{E}(\varphi_t)
		&\geq \frac{1}{2} \int \mathrm{d}x\, |\nabla\varphi_t(x)|^2 - C\|V_2\|_{L^2}\|\varphi_t\|_{L^2}^{5/2}(\varepsilon \|\varphi_t\|_{H^1}^2 + \varepsilon^{-3}) - \frac{1}{4} \| V_{\infty} \|_{L^{\infty}} \| \varphi_t \|_{L^2}^4 \notag\\
		&\geq (1/2 - C\varepsilon\| V_2 \|_{L_2})\|\varphi_t \|_{H^1}^2 - C(\varepsilon^{-3}\| V_2 \|_{L_2} + \| V_{\infty} \|_{L^{\infty}} + 1).\notag
		\end{align}
		Choosing $\varepsilon>0$ small enough, from the energy conservation we get $\|\varphi_t\|_{H^1}<C(\varepsilon, \| V_2 \|_{L_2}, \| V_{\infty} \|_{L^{\infty}})$ for some constant $C$ independent of $t$. This proves the desired lemma.
	\end{proof}
	
	In the conventional Strichartz's estimate, if $u$ is a soulution of
	\[
	\mathrm{i}\partial_t u = -\Delta u + F
	\]
	with initial data $u_0$, then
	\[
	\|u\|_{L^q((0,T),W^{s,r})}\leq C_1\|u_0\|_{H^s}+C_2\|F\|_{L^{\tilde{q}'}((0,T),W^{s,\tilde{r}'})},
	\]
	for admissible pairs $(q,r)$ and $(\tilde{q},\tilde{r})$. (Here, $p'$ denotes the H\"older conjugate of $p$.) For the detail of the Strichartz's estimate, we refer to \cite[pp.73-77]{tao2006nonlinear}.
	
	In this paper, we consider the case $F=(V*|u|^2) u$ and use the following version of Strichartz's Estimate.
	
	\begin{lem}\label{lem:decaying}
		Suppose that $V\in L^2(\mathbb{R}^3)+L^\infty(\mathbb{R}^3)$. Let $\varphi_{t}$ be the solution of the Hartree equation \eqref{eq:Hartree_eq} with initial data $\varphi_{0}=\varphi\in H^{1}(\mathbb{R}^3)$, then { there exists a constant $C$, depending only on $\left\Vert \varphi\right\Vert _{H^{1}}$, $\|V_2\|_{L^2}$, and $\|V_\infty\|_{L^\infty}$, such that}
		\[\|\varphi_t\|_{L^2((0,T),L^\infty)} \leq { C (1+ {{T}} )}.\]
	\end{lem}
	
	\begin{proof}
		We closely follow \cite[Theorem 2.3.3]{cazenave2003semilinear} for the proof of the lemma. Let $V=V_2 + V_\infty$ as in the proof of Lemma \ref{lem:H^1 bound}. From the Sobolev inequality and the Strichartz's estimate,
		we have
		\begin{align}
		\|\varphi_t\|_{L^2((0,T),L^\infty)} &\leq C \|\varphi_t\|_{L^2((0,T),W^{1,6})} \\
		&\leq C\|\varphi_0\|_{H^1} + C\|(V_2*|\varphi_t|^2)\varphi_t\|_{L^2((0,T),W^{1,6/5})} +C\|(V_\infty*|\varphi_t|^2)\varphi_t\|_{L^{{1}}((0,T),W^{1,2})}.
		\label{eq:Strischartz_div}
		\end{align}
		From the definition of the Sobolev norm,
		\begin{align}	
		&\|(V_2*|\varphi_t|^2)\varphi_t\|_{L^2((0,T),W^{1,6/5})} \label{eq:(V_2*phi^2)phi} \\
		&\leq C \|(V_2*|\varphi_t|^2)\varphi_t\|_{L^2((0,T),{L^{6/5}})} + C \|\nabla ((V_2*|\varphi_t|^2)\varphi_t)\|_{L^2((0,T),{L^{6/5}})} \notag
		\end{align}
		and
		\begin{equation}
		\|(V_\infty*|\varphi_t|^2)\varphi_t\|_{L^{{1}}((0,T),W^{1,2})} \leq \|(V_\infty*|\varphi_t|^2)\varphi_t\|_{L^{{1}}((0,T),L^{2})} + \|\nabla ((V_\infty*|\varphi_t|^2)\varphi_t)\|_{L^{{1}}((0,T),L^{2})}\label{eq:(V_infty*phi^2)phi}
		\end{equation}
		We first focus on the spacial integral; integration with respect to the time variable $t$ will be considered later. In the first term in the right-hand side of \eqref{eq:(V_2*phi^2)phi}, the integrand of the spatial integral is bounded by
		\begin{equation}
		\begin{aligned}
		&\|(V_2*|\varphi_t|^2)\varphi_t\|_{L^{6/5}}
		\leq \|V_2*|\varphi_t|^2\|_{L^3}\|\varphi_t\|_{L^2} \leq \|V_2\|_{L^2}\||\varphi_t|^2\|_{L^{6/5}}\|\varphi_t\|_{L^2}  \\
		&\qquad\leq  \|V_2\|_{L^2}\|\varphi_t\|_{L^{12/5}}^2\|\varphi_t\|_{L^2}
		\leq  \|V_2\|_{L^2}\|\varphi_t\|_{L^{2}}^{5/2}\|\varphi_t\|_{L^6}^{1/2}
		\leq \|V_2\|_{L^2}\|\varphi_t\|_{L^2}^{5/2}\|\varphi_t\|_{H^1}^{1/2},
		\end{aligned}
		\end{equation}
		where we used H\"older's inequality, Young's inequality, and Riesz-Thorin Theorem. Similarly, we decompose the integrand of the second term in the right-hand side \eqref{eq:(V_2*phi^2)phi} into two parts and find that
		\[
		\|\nabla ((V_2*|\varphi_t|^2)\varphi_t)\|_{L^2((0,T),{L^{6/5}})}\leq \| (V_2*(\nabla|\varphi_t|^2))\varphi_t\|_{L^2((0,T),{L^{6/5}})} +\| (V_2*|\varphi_t|^2)(\nabla\varphi_t)\|_{L^2((0,T),{L^{6/5}})}.
		\]
		We again apply H\"older's inequality, Young's inequality, and Riesz-Thorin Theorem to get
		\begin{align*}
		\| V_2*(\nabla |\varphi_t|^2 ) \varphi_t\|_{L^{6/5}} &\leq \| V_2*(\nabla|\varphi_t|^2)\|_{L^{3}}\|\varphi_t\|_{L^2}
		\leq  C \|V_2\|_{L^2}\|\overline{\varphi_t}\nabla\varphi_t\|_{L^{6/5}}\|\varphi_t\|_{L^2}\\
		&\leq  C \|V_2\|_{L^2}\|\varphi_t\|_{L^{3}}\|\nabla\varphi_t\|_{L^{2}}\|\varphi_t\|_{L^2}
		\leq C \|V_2\|_{L^2}\|\varphi_t\|_{L^{2}}^{3/2}\|\varphi_t\|_{H^1}^{3/2}, \\
		\| (V_2*|\varphi_t|^2)(\nabla\varphi)\|_{L^{6/5}}
		&\leq \|V_2*|\varphi_t|^2\|_{L^{3}}\|\nabla\varphi_t\|_{L^2}
		\leq \|V_2\|_{L^2}\|\varphi_t\|_{L^2}^{3/2}\|\varphi_t\|_{H^1}^{3/2
		}.
		\end{align*}
		Let us now investigate \eqref{eq:(V_infty*phi^2)phi} using the similar strategy used for \eqref{eq:(V_2*phi^2)phi}. The first term in the right-hand side of \eqref{eq:(V_infty*phi^2)phi} is easily bounded by 
		{
			\begin{align*}
			\| V_\infty*(|\varphi_t|^2 ) \varphi_t\|_{L^{2}} &\leq  \|V_\infty\|_{L^\infty}\||\varphi_t|^2\|_{L^1}\|\varphi_t\|_{L^2} \leq \|V_\infty\|_{L^\infty}\|\varphi_t\|_{L^2}^3.
			\end{align*}
		}
		For the second term, we again split \eqref{eq:(V_infty*phi^2)phi} such that
		\begin{align*}
		\|\nabla ((V_\infty*|\varphi_t|^2)\varphi_t)\|_{L^{{1}}((0,T),L^{2})}&\leq \| (V_\infty*(\nabla|\varphi_t|^2))\varphi_t\|_{L^{{1}}((0,T),L^{2})} \\
		&\qquad +\| (V_\infty*|\varphi_t|^2)(\nabla\varphi_t)\|_{L^{{1}}((0,T),L^
			{2})},
		\end{align*}
		where the right-hand side can be bounded by
		\begin{align*}
		\| V_\infty*(\nabla |\varphi_t|^2 ) \varphi_t\|_{L^{2}} &\leq  \|V_\infty\|_{L^\infty}\|\nabla|\varphi_t|^2\|_{L^1}\|\varphi_t\|_{L^2}\\
		&\leq C \|V_\infty\|_{L^\infty}\|\nabla\varphi_t\|_{L^2}\|\varphi_t\|_{L^2}^2
		\leq C \|V_\infty\|_{L^\infty}\|\varphi_t\|_{H^1}\|\varphi_t\|_{L^2}^2
		\end{align*}
		and
		{
			\[
			\| (V_\infty*|\varphi_t|^2)(\nabla\varphi_t)\|_{L^{2}}
			\leq \|V_\infty\|_{L^\infty}\||\varphi_t|^2\|_{L^1}\|\nabla\varphi_t\|_{L^2}
			\leq \|V_\infty\|_{L^\infty}\|\varphi_t\|_{L^2}^2\|\varphi_t\|_{H^1}.
			\]
		}
		Thus, after taking $L^2$-norm {{and $L^1$-norm appropriately according to \eqref{eq:Strischartz_div}} with respect to the time variable $t$}, with the mass conservation $\|\varphi_t\|_{L^2}=1$ and Lemma \ref{lem:H^1 bound}, we conclude that
		\[
		\|\varphi_t\|_{L^2((0,T),L^\infty)} \leq  C (1 + {{T}}).
		\]
	\end{proof}

	\subsection{Fock space representation}
	
	To analyze the dynamics of the system of $N$-Bosons, we translate the original problem into the language of Fock space as in \cite{Rodnianski2009,Chen2011a}. The Bosonic Fock space we consider is the Hilbert space
	\[
	\mathcal{F}=\bigoplus_{n\geq0}L^{2}\left(\mathbb{R}^{3},\mathrm{d}x\right)^{\otimes_{s}n}=\mathbb{C}\oplus\bigoplus_{n\geq1}L_{s}^{2}\left(\mathbb{R}^{3n},\mathrm{d}x_{1},{{\dots}},\mathrm{d}x_{n}\right),
	\]
	where $L_{s}^{2} = L_{s}^{2}\left(\mathbb{R}^{3n},\mathrm{d}x_{1},{{\dots}},\mathrm{d}x_{n}\right)$ is a subspace of $L^2(\mathbb{R}^{3n},\mathrm{d}x_1,\dots,\mathrm{d}x_n)$ that consists of all functions symmetric under any permutation of $x_1, x_2, \dots, x_n$. By definition, we let $L_{s}^{2}\left(\mathbb{R}^{3}\right)^{\otimes_{s}0}=\mathbb{C}$. An element $\psi \in \mathcal{F}$ can be written as a sequence $\psi=\{\psi^{\left(n\right)}\}_{n\geq0}$ with $n$-particle wave functions $\psi^{\left(n\right)}\in L_{s}^{2}\left(\mathbb{R}^{3n}\right)$. The inner product on $\mathcal{F}$ is defined by
	\[
	\begin{aligned}
	\left\langle \psi_{1},\psi_{2}\right\rangle &
	=\sum_{n\geq0}\langle \psi_{1}^{\left(n\right)},\psi_{2}^{\left(n\right)}\rangle _{L^{2}\left(\mathbb{R}^{3n}\right)}\\
	& =\overline{ \psi_{1}^{(0)} } \psi_{2}^{(0)}+\sum_{n\geq0}\int \mathrm{d}x_{1}{{\dots}} \mathrm{d}x_{n}\overline{\psi_{1}^{\left(n\right)}\left(x_{1},{{\dots}},x_{n}\right)}\psi_{2}^{\left(n\right)}\left(x_{1},{{\dots}},x_{n}\right).
	\end{aligned}
	\]
	The vector $\left\{ 1,0,0,{{\dots}}\right\} \in\mathcal{F}$ is called the vacuum and denoted by $\Omega$. For $f\in L^{2}\left(\mathbb{R}^{3}\right)$, we define the creation operator $a^{*}\left(f\right)$ and the annihilation operator $a\left(f\right)$ on $\mathcal{F}$ by
	\begin{equation} \label{eq:creation}
	\left(a^{*}\left(f\right)\psi\right)^{\left(n\right)}\left(x_{1},{{\dots}},x_{n}\right)=\frac{1}{\sqrt{n}}\sum_{j=1}^{n}f\left(x_{j}\right)\psi^{\left(n-1\right)}\left(x_{1},{{\dots}},x_{j-1},x_{j+1},{{\dots}},x_{n}\right)
	\end{equation}
	and
	\begin{equation} \label{eq:annihilation}
	\left(a\left(f\right)\psi\right)^{\left(n\right)}\left(x_{1},{{\dots}},x_{n}\right)=\sqrt{n+1}\int \mathrm{d}x\overline{f\left(x\right)}\psi^{\left(n+1\right)}\left(x,x_{1},{{\dots}},x_{n}\right).
	\end{equation}
	By definition, the creation operator $a^{*}\left(f\right)$ is the adjoint of the annihilation operator of $a\left(f\right)$, and in particular, $a^{*}\left(f\right)$ and $a\left(f\right)$ are not self-adjoint. We will use the self-adjoint operator $\phi\left(f\right)$ defined as
	\[
	\phi\left(f\right)=a^{*}\left(f\right)+a\left(f\right).
	\]
	We also use operator-valued distributions $a_{x}^{*}$ and $a_{x}$ satisfying
	\[
	a^{*}\left(f\right)=\int \mathrm{d}x\,f\left(x\right)a_{x}^{*}, \qquad a\left(f\right)=\int \mathrm{d}x\,\overline{f\left(x\right)}a_{x}
	\]
	for any $f\in L^{2}\left(\mathbb{R}^{3}\right)$. The canonical commutation relation between these operators is
	\[
	\left[a\left(f\right),a^{*}\left(g\right)\right]=\left\langle f,g\right\rangle _{L^{2}\left(\mathbb{R}^{3}\right)},\quad\left[a\left(f\right),a\left(g\right)\right]=\left[a^{*}\left(f\right),a^{*}\left(g\right)\right]=0,
	\]
	which also assumes the form
	\[
	\left[a_{x},a_{y}^{*}\right]=\delta\left(x-y\right),\quad\left[a_{x},a_{y}\right]=\left[a_{x}^{*},a_{y}^{*}\right]=0.
	\]
	From the creation operator and the annihilation operator, we can also define other useful operators on $\mathcal{F}$.
	For each nonnegative integer $n$, we introduce the projection operator onto the $n$-particle sector of the Fock space,
	\begin{equation}
	P_n(\psi) := (0, 0, \dots, 0, \psi^{(n)}, 0, \dots)
	\end{equation}
	for $\psi = (\psi^{(0)}, \psi^{(1)}, \dots) \in\mathcal{F}$. For simplicity, with slight abuse of notation, we will use $\psi^{(n)}$ to denote $P_n\psi$.
	The number operator $\mathcal{N}$ is given by
	\begin{equation} \label{eq:number operator}
	\mathcal{N}=\int  dx\,a_{x}^{*}a_{x}
	\end{equation}
	and it satisfies that $\left(\mathcal{N}\psi\right)^{\left(n\right)}=n\psi^{\left(n\right)}$. In general, for an operator $J$ defined on the one-particle sector $L^{2}\left(\mathbb{R}^{3},\mathrm{d}x\right)$,
	its second quantization $d\Gamma\left(J\right)$ is the operator on $\mathcal{F}$ whose action on the $n$-particle sector is given by
	\[
	\left(d\Gamma\left(J\right)\psi\right)^{\left(n\right)}=\sum_{j=1}^{n}J_{j}\psi^{\left(n\right)}
	\]
	where $J_{j}=1\otimes{{\dots}}\otimes J\otimes{{\dots}}\otimes1$ is the operator $J$ acting on the $j$-th variable only. The number operator defined above can also be understood as the second quantization
	of the identity, i.e., $\mathcal{N}=d\Gamma\left(1\right)$. With a kernel $J\left(x;y\right)$ of the operator $J$, the second quantization $d\Gamma\left(J\right)$ can be also be written as
	\[
	d\Gamma\left(J\right)=\int \mathrm{d}x \mathrm{d}y\,J\left(x;y\right)a_{x}^{*}a_{y},
	\]
	which is consistent with \eqref{eq:number operator}.
	
	The following lemma shows how to control the annihilation operator and the creation operator (and also the second quantization operators) in terms of the number operator $\mathcal{N}$.
	\begin{lem}[Lemma 3.1 in \cite{Chen2011a}]
		For $\alpha >0$, let $D(\mathcal{N}^{\alpha}) = \{ \psi \in \mathcal{F} : \sum_{n \geq 1} n^{2\alpha} \| \psi^{(n)} \|^2 < \infty \}$ denote the domain of the operator $\mathcal{N}^{\alpha}$. For any $f \in L^2 (\mathbb{R}^3, dx)$ and any $\psi \in D (\mathcal{N}^{1/2})$, we have
		\begin{equation}\label{eq:bd-a}
		\begin{split}  \| a(f) \psi \| & \leq \| f \| \, \| \mathcal{N}^{1/2} \psi \|, \\
		\| a^* (f) \psi \| &\leq \| f \| \, \| (\mathcal{N}+1)^{1/2} \psi \|, \\
		\| \phi (f) \psi \| &\leq 2 \| f \| \| \left( \mathcal{N} + 1 \right)^{1/2}
		\psi \| \, . \end{split} \end{equation}
		Moreover, for any bounded one-particle operator $J$ on $L^2 (\mathbb{R}^3, dx)$ and for every $\psi \in D (\mathcal{N})$, we find
		\begin{equation}
		\label{eq:J-bd} \| d\Gamma (J) \psi \| \leq \| J \| \| \mathcal{N} \psi \|  \, .
		\end{equation}
	\end{lem}
	
	To consider the problem with the Fock space formalism, we extend the Hamiltonian in \eqref{eq:N_body_Hamiltonian} to the Fock space by
	\begin{equation} \label{eq:Fock_space_Hamiltonian}
	\mathcal{H}_{N}:=\int \mathrm{d}x\,a_{x}^{*}\left(-\Delta_{x}\right)a_{x}+\frac{1}{2N}\int \mathrm{d}x \mathrm{d}y\,V\left(x-y\right)a_{x}^{*}a_{y}^{*}a_{y}a_{x}
	\end{equation}
	With the definition, we have $(\mathcal{H}_{N}\psi)^{(N)}=H_{N}\psi^{(N)}$ for $\psi\in\mathcal{F}$.
	The one-particle marginal density $\gamma_{\psi}^{\left(1\right)}$ associated with $\psi$ is
	\begin{equation} \label{eq:Kernel_gamma}
	\gamma_{\psi}^{\left(1\right)}\left(x;y\right)=\frac{1}{\left\langle \psi,\mathcal{N}\psi\right\rangle }\left\langle \psi,a_{y}^{*}a_{x}\psi\right\rangle.
	\end{equation}
	Note that $\gamma_{\psi}^{\left(1\right)}$ is a trace class operator on $L^{2}\left(\mathbb{R}^{3}\right)$ and $\text{Tr }\gamma_{\psi}^{\left(1\right)}=1$. It can be easily checked that \eqref{eq:Kernel_gamma} is equivalent to \eqref{eq:Kernel_of_Marginal_Density}.
	
	Heuristically, if $\psi = \psi^{(N)}\in\mathcal{F}$ were an eigenvector of $a_x$ with the eigenvalue $\sqrt N \varphi(x)$, then from \eqref{eq:Kernel_gamma} we get $\gamma_{\psi}^{\left(1\right)}(x;y) = \varphi(x) \overline{\varphi(y)}$, hence it would coincide with the one-particle marginal density associated with the factorized wave function $\varphi^{\otimes N}$. Even though the eigenvectors of the annihilation operator do not reside in a single sector of the Fock space, they still play the most important role in the midway. These eigenvectors are known as the coherent states, defined by
	\begin{equation}
	\psi\left(f\right)= e^{-\left\| f\right\| ^{2}/2}\sum_{n\geq0}\frac{\left(a^{*}\left(f\right)\right)^{n}}{n!}\Omega=e^{-\left\| f\right\| ^{2}/2}\sum_{n\geq0}\frac{1}{\sqrt{n!}}f^{\otimes n}.
	\end{equation}
	Here, for the ease of notation, when we say a function $\psi^{(N)}\in L^2(\mathbb{R}^{3N})$ is a function {in} the Fock space $\mathcal F$, we mean that $\psi^{(N)}=(0,0,\dots,0,\psi^{(N)},0,\dots)\in\mathcal{F}$. For example, we used $f^{\otimes n}$ to denote
	\begin{equation}
	(0,0,\dots,0,f^{\otimes n},0,\dots)\in\mathcal{F}
	\end{equation}
	whose only nonzero component, $f^{\otimes n}$, is in the $n$-particle sector of the Fock space.
	Closely related to the coherent states is the Weyl operator. For $f\in L^{2}\left(\mathbb{R}^{3}\right)$, the Weyl operator $W\left(f\right)$ is defined by
	\[
	W\left(f\right):=\exp\left(a^{*}\left(f\right)-a\left(f\right)\right)
	\]
	and it also satisfies
	\[
	W\left(f\right)=e^{-\left\Vert f\right\Vert ^{2}/2}\exp\left(a^{*}\left(f\right)\right)\exp\left(-a\left(f\right)\right),
	\]
	which is known as the Hadamard lemma in Lie algebra. The coherent state can also be written in terms of the Weyl operator as
	\begin{equation} \label{Weyl_f}
	\psi\left(f\right)=W\left(f\right)\Omega =e^{-\left\Vert f\right\Vert ^{2}/2}\exp\left(a^{*}\left(f\right)\right)\Omega =e^{-\left\Vert f\right\Vert ^{2}/2}\sum_{n\geq0}\frac{1}{\sqrt{n!}}f^{\otimes n}.
	\end{equation}
	
	We collect the useful properties of the Weyl operator and the coherent states in the following lemma.
	\begin{lem} \label{lem:Basic_Weyl}
		Let $f,g\in L^{2}\left(\mathbb{R}^{3},\mathrm{d}x\right)$.
		\begin{enumerate}
			\item The commutation relation between the Weyl operators is given by
			\[
			W\left(f\right)W\left(g\right)=W\left(g\right)W\left(f\right)e^{-2\mathrm{i} \cdot \mathrm{Im}\left\langle f,g\right\rangle }=W\left(f+g\right)e^{-\mathrm{i} \cdot \mathrm{Im}\left\langle f,g\right\rangle }.
			\]
			
			\item The Weyl operator is unitary and satisfies
			\[
			W\left(f\right)^{*}=W\left(f\right)^{-1}=W\left(-f\right).
			\]
			
			\item The coherent states are eigenvectors of annihilation operators, i.e.,
			\[
			a_{x}\psi\left(f\right)=f\left(x\right)\psi\left(f\right)\quad\Rightarrow\quad a\left(g\right)\psi\left(f\right)=\left\langle g,f\right\rangle _{L^{2}}\psi\left(f\right).
			\]
			The commutation relation between the Weyl operator and the annihilation operator (or the creation operator) is thus
			\[
			W^{*}\left(f\right)a_{x}W\left(f\right)=a_{x}+f\left(x\right)\quad\text{and}\quad W^{*}\left(f\right)a_{x}^{*}W\left(f\right)=a_{x}^{*}+\overline{f\left(x\right)}.
			\]
			
			\item The distribution of $\mathcal{N}$ with respect to the coherent state $\psi\left(f\right)$ is Poisson. In particular,
			\[
			\left\langle \psi\left(f\right),\mathcal{N}\psi\left(f\right)\right\rangle =\left\Vert f\right\Vert ^{2}, \qquad \left\langle \psi\left(f\right),\mathcal{N}^{2}\psi\left(f\right)\right\rangle -\left\langle \psi\left(f\right),\mathcal{N}\psi\left(f\right)\right\rangle ^{2}=\left\Vert f\right\Vert ^{2}.
			\]
		\end{enumerate}
	\end{lem}
	
	We omit the proof of the lemma, since it can be derived from the definition of the Weyl operator and elementary calculation.
	
	Finally, we collect lemmas on the Weyl operator acting on a state with fixed number of particles, which will be used in Section \ref{sec:lemmas}. Define
	\begin{equation} \label{eq:d_N}
	d_{N}:=\frac{\sqrt{N!}}{N^{N/2}e^{-N/2}}.
	\end{equation}
	We note that $C^{-1} N^{1/4} \leq d_N \leq C N^{1/4}$ for some constant $C>0$ independent of $N$, which can be easily checked by using Stirling's formula.
	
	\begin{lem}\label{lem:coherent_all}
		There exists a constant $C>0$ independent of $N$ such that, for any $\varphi\in L^{2}(\mathbb{R}^{3})$ with $\|\varphi\|=1$, we have
		\[
		\left\Vert (\mathcal{N}+1)^{-1/2}W^{*}(\sqrt{N}\varphi)\frac{(a^{*}(\varphi))^{N}}{\sqrt{N!}}\Omega\right\Vert \leq\frac{C}{d_{N}}.
		\]
	\end{lem}
	
	\begin{proof}
		See  \cite[Lemma 6.3]{Chen2011}.
	\end{proof}
	
	\begin{lem} \label{lem:coherent_even_odd}
		Let $P_m$ be the projection onto the $m$-particle sector of the Fock space $\mathcal{F}$ for a non-negative integer $m$. Then, for any non-negative integer $k\leq(1/2)N^{1/3}$,
		\[
		\left\Vert P_{2k}W^{*}(\sqrt{N}\varphi)\frac{(a^{*}(\varphi))^{N}}{\sqrt{N!}}\Omega\right\Vert \leq\frac{2}{d_{N}}
		\]
		and
		\[
		\left\Vert P_{2k+1}W^{*}(\sqrt{N}\varphi)\frac{(a^{*}(\varphi))^{N}}{\sqrt{N!}}\Omega\right\Vert \leq\frac{2(k+1)^{3/2}}{d_{N}\sqrt{N}}.
		\]
	\end{lem}
	
	\begin{proof}
		See \cite[Lemma 7.2]{Lee2013}.
	\end{proof}

	\section{Proof of the Main Result}\label{sec:Pf-of-Main-Thm}
	
	In this section, we prove the main result of the paper, Theorem \ref{thm:Main_Theorem}.
	
	\subsection{Unitary operators and their generators}
	
	Let $\gamma_{N,t}^{(1)}$ be the kernel of the one-particle marginal density associated with the time evolution of the factorized state $\varphi^{\otimes N}$ with respect to the Hamiltonian $\mathcal{H}_N$. By definition,
	\begin{align}
	\gamma_{N,t}^{(1)} & =\frac{\left\langle e^{-\mathrm{i}\mathcal{H}_{N}t}\varphi^{\otimes N},a_{y}^{*}a_{x}e^{-\mathrm{i}\mathcal{H}_{N}t}\varphi^{\otimes N}\right\rangle }{\left\langle e^{-\mathrm{i}\mathcal{H}_{N}t}\varphi^{\otimes N},\mathcal{N}e^{-\mathrm{i}\mathcal{H}_{N}t}\varphi^{\otimes N}\right\rangle } =\frac{1}{N}\left\langle \varphi^{\otimes N}, e^{\mathrm{i}\mathcal{H}_{N}t} a_{y}^{*}a_{x}e^{-\mathrm{i}\mathcal{H}_{N}t}\varphi^{\otimes N}\right\rangle \notag \\
	& =\frac{1}{N}\left\langle \frac{\left(a^{*}(\varphi)\right)^{N}}{\sqrt{N!}}\Omega,e^{\mathrm{i}\mathcal{H}_{N}t}a_{y}^{*}a_{x}e^{-\mathrm{i}\mathcal{H}_{N}t}\frac{\left(a^{*}(\varphi)\right)^{N}}{\sqrt{N!}}\Omega\right\rangle. \label{eq:marginal_factorized}
	\end{align}
	If we use the coherent states instead of the factorized state in \eqref{eq:marginal_factorized} and expand $a_{y}^{*}a_{x}$ around $N \overline{\varphi_t(y)} \varphi_t(x)$, then we are lead to consider the operator
	\begin{align} \label{eq:introducing U}
	&W^{*}(\sqrt{N}\varphi_{s}) e^{\mathrm{i}\mathcal{H}_{N}\left(t-s\right)} (a_x - \sqrt{N} \varphi_t(x) ) e^{-\mathrm{i}\mathcal{H}_{N}\left(t-s\right)} W(\sqrt{N}\varphi_{s}) \\
	&= W^{*}(\sqrt{N}\varphi_{s}) e^{\mathrm{i}\mathcal{H}_{N}\left(t-s\right)} W(\sqrt{N}\varphi_{t}) a_x W^{*}(\sqrt{N}\varphi_{t}) e^{-\mathrm{i}\mathcal{H}_{N}\left(t-s\right)} W(\sqrt{N}\varphi_{s}). \notag
	\end{align}
	To understand further the operator $W^{*}(\sqrt{N}\varphi_{t}) e^{-\mathrm{i}\mathcal{H}_{N}\left(t-s\right)} W(\sqrt{N}\varphi_{s})$, we differentiate it with respect to $t$ as in \cite{Rodnianski2009,Chen2011a,Lee2013}. Direct computation shows that
	\begin{equation} \label{eq:derivative decomposition}
	\mathrm{i} \partial_t W^{*}(\sqrt{N}\varphi_{t}) e^{-\mathrm{i}\mathcal{H}_{N}\left(t-s\right)} W(\sqrt{N}\varphi_{s}) =: \left( \sum_{k=0}^{4} \mathcal{L}_k (t) \right) W^{*}(\sqrt{N}\varphi_{t}) e^{-\mathrm{i}\mathcal{H}_{N}\left(t-s\right)} W(\sqrt{N}\varphi_{s}),
	\end{equation}
	where $\mathcal{L}_k$ is with $k$ creation and/or annihilation operators. The exact formulas for $\mathcal{L}_k$ are as follows:
	\begin{align}
	{{\mathcal{L}_0 (t;s)}} &:= \frac{N}{2} \int_s^t \mathrm{d}\tau \int \mathrm{d}x (V *| \varphi_{\tau}|^2 ) (x) |\varphi_{\tau} (x)|^2, \\
	\mathcal{L}_1 (t) &= 0, \\
	\mathcal{L}_{2}(t) &:= \int \mathrm{d} x\, a_{x}^{*}\left(-\Delta\right)a_{x}+\int \mathrm{d} x\,\left(V*\left|\varphi_{t}\right|^{2}\right)\left(x\right)a_{x}^{*}a_{x} +\int \mathrm{d} x \mathrm{d} y\, V\left(x-y\right)\overline{\varphi_{t}\left(x\right)}\varphi_{t}\left(y\right)a_{y}^{*}a_{x} \notag \\
	& \qquad +\frac{1}{2}\int \mathrm{d} x \mathrm{d} y\, V\left(x-y\right)\left(\varphi_{t}\left(x\right)\varphi_{t}\left(y\right)a_{x}^{*}a_{y}^{*}+\overline{\varphi_{t}\left(x\right)}\,\overline{\varphi_{t}\left(y\right)}a_{x}a_{y}\right), \label{eq:L_2} \\
	\mathcal{L}_{3}(t) &:= \frac{1}{\sqrt{N}}\int \mathrm{d} x \mathrm{d} y\, V\left(x-y\right)\left(\varphi_{t}\left(y\right)a_{x}^{*}a_{y}^{*}+\overline{\varphi_{t}\left(y\right)}a_{x}^{*}a_{y}\right)a_{x} \label{eq:L_3} \\
	\mathcal{L}_{4} &:= \frac{1}{2N}\int \mathrm{d} x \mathrm{d} y\, V\left(x-y\right)a_{x}^{*}a_{y}^{*}a_{x}a_{y}. \label{eq:L_4}
	\end{align}
	Note that {{$\mathcal{L}_0 (t;s)$}} is not an operator but a complex-valued function, which we call the phase factor. Although this term contains the factor $N$, we can cancel this term by multiplying the right-hand side of \eqref{eq:derivative decomposition} by a function {{$e^{-\mathrm{i} \mathcal{L}_0 (t;s)/2}$}} (see Section 3 of \cite{Lee2013}). Thus, if we define the unitary operator $\mathcal{U}(t;s)$ by
	\begin{equation}
	\mathcal{U} (t;s) := e^{-\mathrm{i}{\mathcal{L}_0(t;s)/2}} W^{*}(\sqrt{N}\varphi_{t}) e^{-\mathrm{i}\mathcal{H}_{N}\left(t-s\right)} W(\sqrt{N}\varphi_{s}){,}
	\end{equation}
	then
	\begin{equation}\label{eq:def_mathcalU}
	\mathrm{i}\partial_{t}\mathcal{U}\left(t;s\right)=\left(\mathcal{L}_{2}+\mathcal{L}_{3}+\mathcal{L}_{4}\right)\mathcal{U}\left(t;s\right)\quad\text{and}\quad\mathcal{U}\left(s;s\right)=I
	\end{equation}
	and
	\begin{equation}
	W^{*}(\sqrt{N}\varphi_{s})e^{\mathrm{i}\mathcal{H}_{N}\left(t-s\right)}\left(a_{x}-\sqrt{N}\varphi_{t}\left(x\right)\right)e^{-\mathrm{i}\mathcal{H}_{N}\left(t-s\right)}W(\sqrt{N}\varphi_{s})=\mathcal{U}^{*}\left(t;s\right)a_{x}\mathcal{U}\left(t;s\right).
	\end{equation}
	
	Let $\tilde{\mathcal{L}}=\mathcal{L}_{2}+\mathcal{L}_{4}$ and define the unitary operator $\tilde{\mathcal{U}}\left(t;s\right)$ by
	\begin{equation}\label{eq:def_mathcaltildeU}
	\mathrm{i}\partial_{t}\tilde{\mathcal{U}}\left(t;s\right)=\tilde{\mathcal{L}}\left(t\right)\tilde{\mathcal{U}}\left(t;s\right)\quad\text{and } \quad\tilde{\mathcal{U}}\left(s;s\right)=1.
	\end{equation}
	Since $\tilde{\mathcal{L}}$ does not change the parity of the number of particles,
	\begin{equation}
	\left\langle \Omega,\tilde{\mathcal{U}}^{*}\left(t;0\right)a_{y}\tilde{\mathcal{U}}\left(t;0\right)\Omega\right\rangle =\left\langle \Omega,\tilde{\mathcal{U}}^{*}\left(t;0\right)a_{x}^{*}\tilde{\mathcal{U}}\left(t;0\right)\Omega\right\rangle =0\label{eq:Parity_Consevation}
	\end{equation}
	We refer to Lemma 8.2 in \cite{Lee2013} for a rigorous proof of \eqref{eq:Parity_Consevation}.

	\subsection{Proof of Theorem \ref{thm:Main_Theorem}}
	
	As explained in Section \ref{sec:intro}, we use the technique developed in \cite{Lee2013} to prove Theorem \ref{thm:Main_Theorem}. The proof of Theorem \ref{thm:Main_Theorem} consists of the following two propositions. Recall the definition of $d_N$ in \eqref{eq:d_N}. {In this section, we use $\mathcal{U}(t)$ instead of $\mathcal{U}(t;s)$ for notational simplicity.}
	
	\begin{prop} \label{prop:Et1}
		Suppose that the assumptions in Theorem \ref{thm:Main_Theorem} hold. For a Hermitian operator $J$ on $L^{2}(\mathbb{R}^{3})$, let
		\[
		E_{t}^{1}(J):=\frac{d_{N}}{N}\left\langle W^{*}(\sqrt{N}\varphi)\frac{(a^{*}(\varphi))^{N}}{\sqrt{N!}}\Omega,\mathcal{U}^{*}(t)d\Gamma(J)\mathcal{U}(t)\Omega\right\rangle
		\]
		Then, there exist constants $C$ and $K$, depending only on $\|V_2\|_2$, $\|V_\infty\|_\infty$, and $\sup_{|s|\leq t}\|\varphi_{s}\|_{H^{1}}$, such that
		\[
		\left|E_{t}^{1}(J)\right|\leq\frac{C\|J\|e^{K{t^{3/2}}}}{N}.
		\]
	\end{prop}
	
	\begin{prop} \label{prop:Et2}
		Suppose that the assumptions in Theorem \ref{thm:Main_Theorem} hold. For a Hermitian operator $J$ on $L^{2}(\mathbb{R}^{3})$, let
		\[
		E_{t}^{2}(J):=\frac{d_{N}}{\sqrt{N}}\left\langle W^{*}(\sqrt{N}\varphi)\frac{(a^{*}(\varphi))^{N}}{\sqrt{N!}}\Omega,\mathcal{U}^{*}(t)\phi(J\varphi_{t})\mathcal{U}(t)\Omega\right\rangle
		\]
		Then, there exist constants $C$ and $K$, depending only on $\|V_2\|_2$, $\|V_\infty\|_\infty$, and $\sup_{|s|\leq t}\|\varphi_{s}\|_{H^{1}}$, such that
		\[
		\left|E_{t}^{2}(J)\right|\leq\frac{C\|J\|e^{K{t^{3/2}}}}{N}.
		\]
	\end{prop}
	
	Proof of Propositions \ref{prop:Et1} and \ref{prop:Et2} will be given later in section \ref{sec:Pf-of-Props}. With Propositions \ref{prop:Et1} and \ref{prop:Et2}, we now prove Theorem \ref{thm:Main_Theorem}.

	\begin{proof}[Proof of Theorem \ref{thm:Main_Theorem}]
		Recall that
		\begin{equation}
		\gamma_{N,t}^{(1)} =\frac{1}{N}\left\langle \frac{\left(a^{*}(\varphi)\right)^{N}}{\sqrt{N!}}\Omega,e^{i\mathcal{H}_{N}t}a_{y}^{*}a_{x}e^{-i\mathcal{H}_{N}t}\frac{\left(a^{*}(\varphi)\right)^{N}}{\sqrt{N!}}\Omega\right\rangle.
		\end{equation}
		From the definition of the creation operator in \eqref{eq:creation}, we can easily find that
		\begin{equation} \label{eq:coherent_vec}
		\{0,0,{{\dots}},0,\varphi^{\otimes N},0,{{\dots}}\}=\frac{\left(a^{*}(\varphi)\right)^{N}}{\sqrt{N!}}\Omega,
		\end{equation}
		where the $\varphi^{\otimes N}$ on the left-hand side is in the $N$-th sector of the Fock space. Recall that $P_N$ is the projection onto the $N$-particle sector of the Fock space. From \eqref{Weyl_f}, we find that
		\[
		\frac{\left(a^{*}(\varphi)\right)^{N}}{\sqrt{N!}}\Omega=\frac{\sqrt{N!}}{N^{N/2}e^{-N/2}}P_{N}W(\sqrt{N}\varphi)\Omega=d_{N}P_{N}W(\sqrt{N}\varphi)\Omega.
		\]
		Since $\mathcal{H}_N$ does not change the number of paritcles, we also have that
		\begin{align*}
		\gamma_{N,t}^{(1)}(x;y) & =\frac{1}{N}\left\langle \frac{\left(a^{*}(\varphi)\right)^{N}}{\sqrt{N!}}\Omega,e^{\mathrm{i}\mathcal{H}_{N}t}a_{y}^{*}a_{x}e^{-\mathrm{i}\mathcal{H}_{N}t}\frac{\left(a^{*}(\varphi)\right)^{N}}{\sqrt{N!}}\Omega\right\rangle \\
		& =\frac{d_{N}}{N}\left\langle \frac{\left(a^{*}(\varphi)\right)^{N}}{\sqrt{N!}}\Omega,e^{\mathrm{i}\mathcal{H}_{N}t}a_{y}^{*}a_{x}e^{-\mathrm{i}\mathcal{H}_{N}t}P_{N}W(\sqrt{N}\varphi)\Omega\right\rangle \\
		& =\frac{d_{N}}{N}\left\langle \frac{\left(a^{*}(\varphi)\right)^{N}}{\sqrt{N!}}\Omega,P_{N}e^{\mathrm{i}\mathcal{H}_{N}t}a_{y}^{*}a_{x}e^{-\mathrm{i}\mathcal{H}_{N}t}W(\sqrt{N}\varphi)\Omega\right\rangle \\
		& =\frac{d_{N}}{N}\left\langle \frac{\left(a^{*}(\varphi)\right)^{N}}{\sqrt{N!}}\Omega,e^{\mathrm{i}\mathcal{H}_{N}t}a_{y}^{*}a_{x}e^{-\mathrm{i}\mathcal{H}_{N}t}W(\sqrt{N}\varphi)\Omega\right\rangle. \\
		\end{align*}
		To simplify it further, we use the relation
		\[
		e^{\mathrm{i}\mathcal{H}_{N}t}a_{x}e^{-\mathrm{i}\mathcal{H}_{N}t}=W(\sqrt{N}\varphi)\mathcal{U}^{*}(t)(a_{x}+\sqrt{N}\varphi_{t}(x))\mathcal{U}(t)W^{*}(\sqrt{N}\varphi)
		\]
		(and an analogous result for the creation operator) to obtain that
		\begin{align*}
		\gamma_{N,t}^{(1)}(x;y) & =\frac{d_{N}}{N}\left\langle \frac{\left(a^{*}(\varphi)\right)^{N}}{\sqrt{N!}}\Omega,e^{\mathrm{i}\mathcal{H}_{N}t}a_{y}^{*}a_{x}e^{-\mathrm{i}\mathcal{H}_{N}t}W(\sqrt{N}\varphi)\Omega\right\rangle \\
		& =\frac{d_{N}}{N}\left\langle \frac{\left(a^{*}(\varphi)\right)^{N}}{\sqrt{N!}}\Omega,W(\sqrt{N}\varphi)\mathcal{U}^{*}(t)(a_{y}^{*}+\sqrt{N}\,\overline{\varphi_{t}\left(y\right)})(a_{x}+\sqrt{N}\varphi_{t}(x))\mathcal{U}(t)\Omega\right\rangle.
		\end{align*}
		Thus,
		\begin{align*}
		\gamma_{N,t}^{(1)}(x;y)-\overline{\varphi_{t}\left(y\right)}\varphi(x) & =\frac{d_{N}}{N}\left\langle \frac{\left(a^{*}(\varphi)\right)^{N}}{\sqrt{N!}}\Omega,W(\sqrt{N}\varphi)\mathcal{U}^{*}(t)a_{y}^{*}a_{x}\mathcal{U}(t)\Omega\right\rangle \\
		& \quad+\overline{\varphi_{t}\left(y\right)}\frac{d_{N}}{\sqrt{N}}\left\langle \frac{\left(a^{*}(\varphi)\right)^{N}}{\sqrt{N!}}\Omega,W(\sqrt{N}\varphi)\mathcal{U}^{*}(t)a_{x}\mathcal{U}(t)\Omega\right\rangle \\
		& \quad+\varphi_{t}(x)\frac{d_{N}}{\sqrt{N}}\left\langle \frac{\left(a^{*}(\varphi)\right)^{N}}{\sqrt{N!}}\Omega,W(\sqrt{N}\varphi)\mathcal{U}^{*}(t)a_{y}^{*}\mathcal{U}(t)\Omega\right\rangle.
		\end{align*}
		Recall the definition of $E_{t}^{1}(J)$ and $E_{t}^{2}(J)$ in Propositions \ref{prop:Et1} and \ref{prop:Et2}. For any compact one-particle Hermitian operator $J$ on $L^{2}(\mathbb{R}^{3})$, we have
		\begin{align*}
		\Tr J({\gamma}_{N,t}^{(1)}-\left|{\varphi}_{t}\right\rangle \left\langle {\varphi}_{t}\right|) & =\int\mathrm{d} x\mathrm{d} yJ(x;y)\left(\gamma_{N,t}^{(1)}(y;x)-\varphi_{t}(y)\overline{\varphi_{t}\left(x\right)}\right)\\
		& =\frac{d_{N}}{N}\left\langle \frac{\left(a^{*}(\varphi)\right)^{N}}{\sqrt{N!}}\Omega,W(\sqrt{N}\varphi)\mathcal{U}^{*}(t)d\Gamma(J)\mathcal{U}(t)\Omega\right\rangle \\
		& \quad+\frac{d_{N}}{\sqrt{N}}\left\langle \frac{\left(a^{*}(\varphi)\right)^{N}}{\sqrt{N!}}\Omega,W(\sqrt{N}\varphi)\mathcal{U}^{*}(t)\phi(J\varphi_{t})\mathcal{U}(t)\Omega\right\rangle \\
		& =E_{t}^{1}(J)+E_{t}^{2}(J).
		\end{align*}
		Thus, from Propositions \ref{prop:Et1} and \ref{prop:Et2}, we find that
		\[
		\left|\Tr J({\gamma}_{N,t}^{(1)}-\left|{\varphi}_{t}\right\rangle \left\langle {\varphi}_{t}\right|)\right|\leq C\frac{\left\Vert J\right\Vert }{N}e^{Kt{^{3/2}}}.
		\]
		Since the space of compact operators is the dual to that of the trace class operators, and since ${\gamma}_{N,t}^{(1)}$ and $\left|{\varphi}_{t}\right\rangle \left\langle {\varphi}_{t}\right|$
		are Hermitian,
		\[
		\Tr\left|{\gamma}_{N,t}^{(1)}-\left|{\varphi}_{t}\right\rangle \left\langle {\varphi}_{t}\right|\right|\leq\frac{C}{N}e^{Kt{^{3/2}}}
		\]
		which concludes the proof of Theorem \ref{thm:Main_Theorem}.
	\end{proof}
	
	\section{Comparison of Dynamics and Proof of Propositions \ref{prop:Et1} and \ref{prop:Et2}} \label{sec:lemmas}
	
	In this section, we prove Propositions \ref{prop:Et1} and \ref{prop:Et2}. { Recall that $V = V_2 + V_{\infty}$ with $V_2 \in L^2 (\mathbb{R}^3)$ and $V_{\infty} \in L^{\infty} (\mathbb{R}^3)$. Throughout the section, the constants $C$ and $K$ may depend on $\| \varphi_0 \|_{H^{1}(\mathbb{R}^{3})}$, $\| V_2 \|_{L^2}$, and $\| V_{\infty}\|_{L^{\infty}}$.}
	
	\subsection{Comparison of dynamics} \label{subsec:comparison}
	
	As briefly mentioned in Section \ref{sec:intro}, the key technical estimate is the upper bound on the fluctuation of the expected number of particles under the evolution $\mathcal{U}(t;s)$, which is the following lemma.
	
	\begin{lem}\label{lem:NjU}
		Suppose that the assumptions in Theorem \ref{thm:Main_Theorem} hold. Let $\mathcal{U}\left(t;s\right)$ be the unitary evolution defined in \eqref{eq:def_mathcalU}. Then for any $\psi\in\mathcal{F}$ and $j\in\mathbb{N}$, there exist constants $C \equiv C(j)$ and $K \equiv K(j)$ such that
		\[
		\left\langle \mathcal{U}\left(t;s\right)\psi,\mathcal{N}^{j}\mathcal{U}\left(t;s\right)\psi\right\rangle \leq Ce^{K{t^{3/2}}}\left\langle \psi,\left(\mathcal{N}+1\right)^{2j+2}\psi\right\rangle .
		\]
	\end{lem}
	
	To prove Lemma \ref{lem:NjU}, we closely follow the proof of Proposition 3.3 in \cite{Rodnianski2009}. In the proof, we will frequently use an estimate on $\|V(x-\cdot)\varphi_t\|_{L^2}$. The proof of such an estimate is almost immediate with the condition $V \leq D(1-\Delta)$ by using Hardy's inequality. In this paper, however, due to the singularity of $V$, we use Strichartz's estimate as follows:
	\begin{align*}
	\int\mathrm{d} y\,V^{2}(x-y)|\varphi_{t}(y)|^{2}
	\quad&\leq {2}\int\mathrm{d} y\,V_2^{2}(x-y)|\varphi_{t}(y)|^{2} + {2}\int\mathrm{d} y\,V_\infty^{2}(x-y)|\varphi_{t}(y)|^{2}\\
	\quad&\leq {2}\|\varphi_t\|_{L^\infty}^2\int\mathrm{d} y\,V_2^{2}(x-y) + {2} \|V_\infty\|_{L^\infty}^2\int\mathrm{d} y\,|\varphi_{t}(y)|^{2}\\
	\quad&\leq {2}\|\varphi_t\|_{L^\infty}^2\|V_2\|_{L^2}^2 + {2} \|V_\infty\|_{L^\infty}^2\|\varphi_t\|_{L^2}^2\\
	\quad&\leq {2}\|\varphi_t\|_{L^\infty}^2\|V_2\|_{L^2}^2 + {2} \|V_\infty\|_{L^\infty}^2\|\varphi_t\|_{L^2}^2
	\end{align*}
	hence,
	\[
	\sup_x \|V(x-\cdot)\varphi_t\|_{ L^2}\leq C(\|\varphi_{t}\|_{\infty}+1)
	\]
	for some constant $C>0$. Then, using Lemma \ref{lem:decaying} we find that
	\begin{align} \label{eq:VphiLt2Lxinfty}
	\int_0^t\mathrm{d}s\,\sup_x \|V(x-\cdot)\varphi_s\|_{ L^2} &\leq \int_0^t\mathrm{d}s\, C(\|\varphi_{t}\|_{\infty}+1)
	\leq C(\sqrt{t}\|\varphi_t\|_{L^{2}((0,T),L^\infty)}+t)
	\leq C(1+{t^{3/2}}).
	\end{align}
	We now begin the proof of Lemma \ref{lem:NjU}. First, we introduce a truncated time-dependent generator with fixed $M>0$ as follows:
	\begin{align*}
	\mathcal{L}_{N}^{(M)}(t) & =\int\mathrm{d}xa_{x}^{*}(-\Delta_{x})a_{x}+\int\mathrm{d}x\left(V*|\varphi_{t}|^{2}\right)(x)a_{x}^{*}a_{x}+\int\mathrm{d}x\mathrm{d}y V(x-y)\overline{\varphi_{t}(x)}\varphi_{t}(y)a_{y}^{*}a_{x}\\
	& \quad+\frac{1}{2}\int\mathrm{d}x\mathrm{d}yV(x-y)\left(\varphi_{t}(x)\varphi_{t}(y)a_{x}^{*}a_{y}^{*}+\overline{\varphi_{t}(x)}\overline{\varphi_{t}(y)}a_{x}a_{y}\right)\\
	& \quad+\frac{1}{\sqrt{N}}\int\mathrm{d}x\mathrm{d}yV(x-y)a_{x}^{*}\left(\overline{\varphi_{t}(y)}a_{y}\chi(\mathcal{N}\leq M)+\varphi_{t}(y)\chi(\mathcal{N}\leq M)a_{y}^{*}\right)a_{x}\\
	& \quad+\frac{1}{2N}\int\mathrm{d}x\mathrm{d}yV(x-y)a_{x}^{*}a_{y}^{*}a_{y}a_{x}.
	\end{align*}
	We remark that $M$ will be chosen to be $M=N$ later in the proof of Lemma \ref{lem:NjU}. Define a unitary operator $\mathcal{U}^{(M)}$ by
	\begin{equation}\label{eq:def_mathcalUM}
	\mathrm{i}\partial_{t}\mathcal{U}^{(M)}\left(t;s\right)=\mathcal{L}^{(M)}_N(t)\mathcal{U}^{(M)}(t;s)\quad\text{and}\quad\mathcal{U}^{(M)}\left(s;s\right)=1.
	\end{equation}
	We use a three-step strategy.
	
	\vskip15pt
	
	\noindent\emph{Step 1. Truncation with respect to $\mathcal{N}$ with $M>0$.}
	\begin{lem} \label{lem:truncation}
		Suppose that the assumptions in Theorem \ref{thm:Main_Theorem} hold and {let} $\mathcal{U}^{(M)}$ be {the} unitary operator defined in \eqref{eq:def_mathcalUM}. Then, there exist constants $C$ and $K$ such that, for all $N\in\mathbb{N}$ and $M>0$, $\psi\in\mathcal{F}$, and $t,s,\in\mathbb{R}$,
		\[
		\left\langle \mathcal{U}^{(M)}(t;s)\psi,(\mathcal{N}+1)^{j}\mathcal{U}^{(M)}(t;s)\psi\right\rangle \leq\left\langle \psi,(\mathcal{N}+1)^{j}\psi\right\rangle C \exp\left(4^{j}K|t-s|{^{3/2}}(1+\sqrt{M/N})\right).
		\]
	\end{lem}
	\begin{proof}
		Following the proof of Lemma 3.5 in \cite{Rodnianski2009}, we get
		\begin{align}
		&\frac{\mathrm{d}}{\mathrm{d}t}\left\langle \mathcal{U}^{(M)}(t;0)\psi,(\mathcal{N}+1)^{j}\mathcal{U}^{(M)}(t;0)\psi\right\rangle \notag\\
		&\quad= 2\sum_{k=0}^{j-1}{j \choose k}(-1)^{k}\operatorname{Im}\int\mathrm{d}x\mathrm{d}y\,V(x-y)\varphi_{t}(x)\varphi_{t}(y)\notag\\
		&\quad\qquad\times\left\langle \mathcal{U}^{(M)}(t;0)\psi,\left(\mathcal{N}^{k/2}a_{x}^{*}a_{y}^{*}(\mathcal{N}+2)^{k/2}+(\mathcal{N}+1)^{k/2}a_{x}^{*}a_{y}^{*}(\mathcal{N}+3)^{k/2}\right)\mathcal{U}^{(M)}(t;0)\psi\right\rangle \notag\\
		&\quad\qquad+\frac{2}{\sqrt{N}}\sum_{k=0}^{j-1}{j \choose k}\operatorname{Im}\int\mathrm{d}x\notag\\
		&\quad\qquad\quad\qquad\times\left\langle \mathcal{U}^{(M)}(t;0)\psi,a_{x}^{*}a(V(x-\cdot))\chi(\mathcal{N}\leq M)(\mathcal{N}+1)^{k/2}a_{x}\mathcal{N}^{k/2}\mathcal{U}^{(M)}(t;0)\psi\right\rangle. \label{eq:ddtUM}
		\end{align}
		To control the contribution from the first term in the right-hand side of \eqref{eq:ddtUM}, we use the bounds of the form
		\begin{align*}
		&
		\left|\int\mathrm{d}x\mathrm{d}yV(x-y)\varphi_{t}(x)\varphi_{t}(y)\langle\mathcal{U}^{(M)}(t;0)\psi,(\mathcal{N}+1)^{\frac{k}{2}}a_{x}^{*}a_{y}^{*}(\mathcal{N}+3)^{\frac{k}{2}}\mathcal{U}^{(M)}(t;0)\psi\rangle\right|\\
		&
		\quad\leq\int\mathrm{d}x|\varphi_{t}(x)|\|a_{x}(\mathcal{N}+1)^{\frac{k}{2}}\mathcal{U}^{(M)}(t;0)\psi\|\|a^{*}(V(x-\cdot)\varphi_{t})(\mathcal{N}+3)^{\frac{k}{2}}\mathcal{U}^{(M)}(t;0)\psi\|\\
		&
		\quad\leq K\sup_{x}\|V(x-\cdot)\varphi_{t}\|\|(\mathcal{N}+3)^{\frac{k+1}{2}}\mathcal{U}^{(M)}(t;0)\psi\|^{2}\\
		&
		\quad\leq K\sup_{x}\|V(x-\cdot)\varphi_{t}\|\|(\mathcal{N}+1)^{\frac{k+1}{2}}\mathcal{U}^{(M)}(t;0)\psi\|^{2}.
		\end{align*}
		On the other hand, to control contribution arising from the second integral in the right-hand side of \eqref{eq:ddtUM}, we use that
		\begin{align*}
		&
		\left|\int\mathrm{d}x\mathrm{d}yV(x-y)\varphi_{t}(y)\langle\mathcal{U}^{(M)}(t;0)\psi,a_{x}^{*}a_{y}\chi(\mathcal{N}\leq M)(\mathcal{N}+1)^{\frac{k}{2}}a_{x}\mathcal{N}^{\frac{k}{2}}\mathcal{U}^{(M)}(t;0)\psi\rangle\right|\\
		&
		\quad\leq\int\mathrm{d}x\left\Vert a_{x}(\mathcal{N}+1)^{\frac{k}{2}}\mathcal{U}^{(M)}(t;0)\psi\right\Vert \left\Vert a(V(x-\cdot)\varphi_{t})\chi(\mathcal{N}\leq M)\right\Vert \left\Vert a_{x}\mathcal{N}^{\frac{k}{2}}\mathcal{U}^{(M)}(t;0)\psi\right\Vert\\
		&
		\quad\leq KM^{1/2}\sup_{x}\|V(x-\cdot)\varphi_{t}\|\|(\mathcal{N}+1)^{\frac{k+1}{2}}\mathcal{U}^{(M)}(t;0)\psi\|^{2}.
		\end{align*}
		This implies
		\begin{align*}
		& \frac{\mathrm{d}}{\mathrm{d}t}\left\langle \mathcal{U}^{(M)}(t;0)\psi,(\mathcal{N}+1)^{j}\mathcal{U}^{(M)}(t;0)\psi\right\rangle \\
		& \quad\leq K \sup_{x}\|V(x-\cdot)\varphi_{t}\| \left(1+\sqrt{M/N}\right)\sum_{k=0}^{j}{j \choose k}\left\langle \mathcal{U}^{(M)}(t;0)\psi,(\mathcal{N}+3)^{j}\mathcal{U}^{(M)}(t;0)\psi\right\rangle \\
		& \quad\leq4^{j}K \sup_{x}\|V(x-\cdot)\varphi_{t}\| \left(1+\sqrt{M/N}\right)\left\langle \mathcal{U}^{(M)}(t;0)\psi,(\mathcal{N}+1)^{j}\mathcal{U}^{(M)}(t;0)\psi\right\rangle .
		\end{align*}
		Applying the Gronwall Lemma with \eqref{eq:VphiLt2Lxinfty}, we get the desired result.
	\end{proof}

	\noindent\emph{Step 2: Weak bounds on the $\mathcal{U}$ dynamics.}
	\begin{lem}[Lemma 3.6 in \cite{Rodnianski2009}] \label{lem:3.6 i Rodnianski}
		For arbitrary $t,s\in\mathbb{R}$ and $\psi\in\mathcal{F}$, we have {
			\[
			\left\langle \psi,\mathcal{U}(t;s) \mathcal{N}\mathcal{U}(t;s)^{*}\psi\right\rangle \leq6\left\langle \psi,(\mathcal{N}+N+1)\psi\right\rangle .
			\] }
		Moreover, for every $\ell\in\mathbb{N}$, there exists a constant $C(\ell)$ such that
		\[
		\left\langle \psi,\mathcal{U}(t;s)\mathcal{N}^{2\ell}\mathcal{U}(t;s)^{*}\psi\right\rangle \leq C(\ell)\left\langle \psi,(\mathcal{N}+N)^{2\ell}\psi\right\rangle ,
		\]
		\[
		\left\langle \psi,\mathcal{U}(t;s)\mathcal{N}^{2\ell+1}\mathcal{U}(t;s)^{*}\psi\right\rangle \leq C(\ell)\left\langle \psi,(\mathcal{N}+N)^{2\ell+1}(\mathcal{N}+1)\psi\right\rangle
		\]
		for all $t,s\in\mathbb{R}$ and $\psi\in\mathcal{F}$.
	\end{lem}
	
	\begin{proof}
		We may follow the proof of Lemma 3.6 in \cite{Rodnianski2009} without any change.
	\end{proof}

	\noindent\emph{Step 3: Comparison between the $\mathcal{U}$ and $\mathcal{U}^{(M)}$ dynamics.}
	\begin{lem} \label{lem:comparison}
		Suppose that the assumptions in Theorem \ref{thm:Main_Theorem} hold. Then, for every $j\in\mathbb{N}$, there exist constants $C \equiv C(j)$ and $K \equiv K(j)$ such that
		\begin{align*}
		& \left|\left\langle \mathcal{U}(t;s)\psi,\mathcal{N}^{j}\left(\mathcal{U}(t;s)-\mathcal{U}^{(M)}(t;s)\right)\psi\right\rangle \right|\\
		& \quad\leq C(j)(N/M)^j \|(\mathcal{N}+1)^{j+1}\psi\|^2\exp\left(K(j)|t-s|{^{3/2}}(1+\sqrt{M/N})\right)
		\end{align*}
		and
		\begin{align*}
		& \left|\left\langle \mathcal{U}^{(M)}(t;s)\psi,\mathcal{N}^{j}\left(\mathcal{U}(t;s)-\mathcal{U}^{(M)}(t;s)\right)\psi\right\rangle \right|\\
		& \quad\leq C(j)(1/M)^j\|(\mathcal{N}+1)^{j+1}\psi\|^2 \exp\left(K(j)|t-s|{^{3/2}}(1+\sqrt{M/N})\right){.}
		\end{align*}
	\end{lem}
	\begin{proof}
		To simplify the notation we consider the case $s=0$ and $t>0$ only; other cases can be treated in a similar manner. To prove the first inequality of the lemma, we expand the difference of the two evolutions as follows: {
			\begin{align*}
			& \left\langle \mathcal{U}(t;0)\psi,\mathcal{N}^{j}\left(\mathcal{U}(t;0)-\mathcal{U}^{(M)}(t;0)\right)\psi\right\rangle = \left\langle \mathcal{U}(t;0)\psi,\mathcal{N}^{j}\mathcal{U}(t;0)\left(1-\mathcal{U}(t;0)^{*}\mathcal{U}^{(M)}(t;0)\right)\psi\right\rangle \\
			& = -\mathrm{i}\int_{0}^{t}\mathrm{d}s\,\left\langle \mathcal{U}(t;s)\psi,\mathcal{N}^{j}\mathcal{U}(t;s)\left(\mathcal{L}_{N}(s)-\mathcal{L}_{N}^{(M)}(s)\right)\mathcal{U}^{(M)}(s;0)\psi\right\rangle \\
			& = -\frac{\mathrm{i}}{\sqrt{N}}\int_{0}^{t}\mathrm{d}s\,\int\mathrm{d}x\mathrm{d}y\,V(x-y)\\
			& \qquad\times\langle\mathcal{U}(t;0)\psi,\mathcal{N}^{j}\mathcal{U}(t;s)a^{*}_x\left(\overline{\varphi_{s}(y)}a_{y}\chi(\mathcal{N}>M)+\varphi_{s}(y)\chi(\mathcal{N}>M)a_{y}^{*}\right)a_{x}\mathcal{U}^{(M)}(s;0)\psi\rangle\\
			& = -\frac{\mathrm{i}}{\sqrt{N}}\int_{0}^{t}\mathrm{d}s\,\int\mathrm{d}x\,\langle a_{x}\mathcal{U}(t;s)^{*}\mathcal{N}^{j}\mathcal{U}(t;0)\psi,a(V(x-\cdot)\varphi_{s})\chi(\mathcal{N}>M)a_{x}\mathcal{U}^{(M)}(s;0)\psi\rangle\\
			& \quad\quad-\frac{\mathrm{i}}{\sqrt{N}}\int_{0}^{t}\mathrm{d}s\,\int\mathrm{d}x\,\langle a_{x}\mathcal{U}(t;s)^{*}\mathcal{N}^{j}\mathcal{U}(t;0)\psi,\chi(\mathcal{N}>M)a^*(V(x-\cdot)\varphi_{s})a_{x}\mathcal{U}^{(M)}(s;0)\psi\rangle
			\end{align*} }
		Hence,
		\begin{align}
		& \left|\left\langle \mathcal{U}(t;0)\psi,\mathcal{N}^{j}\left(\mathcal{U}(t;0)-\mathcal{U}^{(M)}(t;0)\right)\psi\right\rangle \right|\notag\\
		& \leq\frac{C}{\sqrt{N}}\int_{0}^{t}\mathrm{d}s\,\int\mathrm{d}x\,\|a_{x}\mathcal{U}(t;s)^{*}\mathcal{N}^{j}\mathcal{U}(t;0)\psi\|\cdot\|a(V(x-\cdot)\varphi_{s})a_{x}\chi(\mathcal{N}>M+1)\mathcal{U}^{(M)}(s;0)\psi\|\notag\\
		& \quad+\frac{C}{\sqrt{N}}\int_{0}^{t}\mathrm{d}s\,\int\mathrm{d}x\,\|a_{x}\mathcal{U}(t;s)^{*}\mathcal{N}^{j}\mathcal{U}(t;0)\psi\|\cdot\|a^*(V(x-\cdot)\varphi_{s})a_{x}\chi(\mathcal{N}>M)\mathcal{U}^{(M)}(s;0)\psi\|\notag\\
		& \leq\frac{C}{\sqrt{N}}\int_{0}^{t}\mathrm{d}s\,\sup_{x}\|V(x-\cdot)\varphi_{s}\|\int\mathrm{d}x\,\|a_{x}\mathcal{U}(t;s)^{*}\mathcal{N}^{j}\mathcal{U}(t;0)\psi\|\cdot\|a_{x}(\mathcal{N}+1)^{1/2}\chi(\mathcal{N}>M)\mathcal{U}^{(M)}(s;0)\psi\|\notag\\
		& \leq { \frac{C}{\sqrt{N}}} \int_{0}^{t}\mathrm{d}s\,\sup_{x}\|V(x-\cdot)\varphi_{s}\|\cdot\|(\mathcal{N}+1)^{1/2}\mathcal{U}(t;s)^{*}\mathcal{N}^{j}\mathcal{U}(t;0)\psi\|\cdot\|(\mathcal{N}+1)\chi(\mathcal{N}>M)\mathcal{U}^{(M)}(s;0)\psi\|. \label{eq:U-UM}
		\end{align}
		{
			From Lemma \ref{lem:3.6 i Rodnianski},
			\begin{align}
			\|\mathcal{N}^{1/2}\mathcal{U}(t;s)^{*}\mathcal{N}^{j}\mathcal{U}(t;0)\psi\|^2 &\leq 6 \langle \mathcal{N}^{j}\mathcal{U}(t;0)\psi, (\mathcal{N}+N+1)\mathcal{N}^{j}\mathcal{U}(t;0)\psi \rangle \\
			&\leq C(j) \langle \psi, (\mathcal{N}+N)^{2j+1} (\mathcal{N}+1)\psi \rangle \leq C(j) N^{2j+1} \langle \psi, (\mathcal{N}+1)^{2j+2} \psi \rangle. \notag
			\end{align}
			Since $\chi(\mathcal{N}>M)\leq(\mathcal{N}/M)^{2j}$, we find that
			\begin{align*}
			& \left|\left\langle \mathcal{U}(t;0)\psi,\mathcal{N}^{j}\left(\mathcal{U}(t;0)-\mathcal{U}^{(M)}(t;0)\right)\psi\right\rangle \right|\\
			& \quad\leq C(j)N^{j}\|(\mathcal{N}+1)^{j+1}\psi\|\int_{0}^{t}\mathrm{d}s\,\sup_{x}\|V(x-\cdot)\varphi_{s}\|\left\langle \mathcal{U}^{(M)}(s;0)\psi,(\mathcal{N}+1)^{2}\chi(\mathcal{N}>M)\mathcal{U}^{(M)}(s;0)\psi\right\rangle^{1/2} \\
			& \quad\leq C(j)N^{j}\|(\mathcal{N}+1)^{j+1}\psi\|\int_{0}^{t}\mathrm{d}s\,\sup_{x}\|V(x-\cdot)\varphi_{s}\|\left\langle \mathcal{U}^{(M)}(s;0)\psi,\frac{(\mathcal{N}+1)^{2j+2}}{M^{2j}}\mathcal{U}^{(M)}(s;0)\psi\right\rangle^{1/2}
			\end{align*}
		}
		We thus conclude that
		\begin{align}
		& \left|\left\langle \mathcal{U}(t;0)\psi,\mathcal{N}^{j}\left(\mathcal{U}(t;0)-\mathcal{U}^{(M)}(t;0)\right)\psi\right\rangle \right| \notag\\
		& \quad\leq C(j)(N/M)^{j}\|(\mathcal{N}+1)^{j+1}\psi\|^{2}\int_{0}^{t}\mathrm{d}s\,\sup_{x}\|V(x-\cdot)\varphi_{s}\|\exp\left(K(j)|t-s|{^{3/2}}(1+\sqrt{M/N})\right)\\
		& \quad\leq C(j)(N/M)^{j}\|(\mathcal{N}+1)^{j+1}\psi\|^{2}\exp\left(K(j)|t|{^{3/2}}(1+\sqrt{M/N})\right). \notag
		\end{align}
		The proof of the second part of the lemma is similar and we omit it.
	\end{proof}
	
	We now prove Lemma \ref{lem:NjU} by combining the three steps above.
	\begin{proof}[Proof of Lemma \ref{lem:NjU}]
		From Lemmas \ref{lem:truncation}, \ref{lem:3.6 i Rodnianski}, and \ref{lem:comparison} with the choice $M=N$,
		\begin{align*}
		&\left\langle {\mathcal{U}}\left(t;s\right)\psi,\mathcal{N}^{j}{\mathcal{U}}\left(t;s\right)\psi\right\rangle \\
		&=\left\langle \mathcal{U}\left(t;s\right)\psi,\mathcal{N}^{j}(\mathcal{U}-\mathcal{U}^{(M)})\left(t;s\right)\psi\right\rangle + \left\langle (\mathcal{U}-\mathcal{U}^{(M)})\left(t;s\right)\psi,\mathcal{N}^{j}\mathcal{U}^{(M)}\left(t;s\right)\psi\right\rangle \\
		&\qquad+\left\langle \mathcal{U}^{(M)}\left(t;s\right)\psi,\mathcal{N}^{j}\mathcal{U}^{(M)}\left(t;s\right)\psi\right\rangle \\
		&\leq Ce^{K|t-s|{^{3/2}}}\left\langle \psi,\left(\mathcal{N}+1\right)^{2j+2}\psi\right\rangle.
		\end{align*}
	\end{proof}
	
	Recall the definition of $\tilde{\mathcal{U}}\left(t;s\right)$ in \eqref{eq:def_mathcaltildeU}. In the next lemma, we prove an estimate similar to Lemma \ref{lem:NjU} for the evolution with respect to $\tilde{\mathcal{U}}$.
	
	\begin{lem}\label{lem:tildeNj}
		Suppose that the assumptions in Theorem \ref{thm:Main_Theorem} hold. Then, for any $\psi\in\mathcal{F}$ and $j\in\mathbb{N}$, there exist constants $C \equiv C(j)$ and $K \equiv K(j)$ such that
		\[
		\left\langle \tilde{\mathcal{U}}\left(t;s\right)\psi,\mathcal{N}^{j}\tilde{\mathcal{U}}\left(t;s\right)\psi\right\rangle \leq Ce^{K|t-s|}\left\langle \psi,\left(\mathcal{N}+1\right)^{j}\psi\right\rangle.
		\]
	\end{lem}
	\begin{proof}
		Let $\tilde{\psi}=\tilde{\mathcal{U}}(t;s)\psi$. We have
		\begin{align*}
		\frac{\mathrm{d}}{\mathrm{d} t}\left\langle \tilde{\psi},(\mathcal{N}+1)^{j}\tilde{\psi}\right\rangle = & \left\langle \tilde{\psi},[\mathrm{i}(\mathcal{L}_{2} +\mathcal{L}_{4}),(\mathcal{N}+1)^{j}]\tilde{\psi}\right\rangle \notag\\
		= & {-} \operatorname{Im}\int\mathrm{d} x\mathrm{d} y\,V(x-y)\varphi_{t}(x)\varphi_{t}(y)\left\langle \tilde{\psi},[a_{x}^{*}a_{y}^{*},(\mathcal{N}+1)^{j}]\tilde{\psi}\right\rangle \notag\\
		= & 2\operatorname{Im}\int\mathrm{d} x\mathrm{d} y\,V(x-y)\varphi_{t}(x)\varphi_{t}(y)\left\langle \tilde{\psi},a_{x}^{*}a_{y}^{*}((\mathcal{N}+1)^{j}-(\mathcal{N}+3)^{j})\tilde{\psi}\right\rangle \notag\\
		= & 2 \operatorname{Im}\int\mathrm{d} x\mathrm{d} y\,V(x-y)\varphi_{t}(x)\varphi_{t}(y)\notag\\
		& \times\left\langle (\mathcal{N}+3)^{j/2-1}a_{x}a_{y}\tilde{\psi},(\mathcal{N}+3)^{1-j/2}((\mathcal{N}+1)^{j}-(\mathcal{N}+3)^{j})\tilde{\psi}\right\rangle .
		\end{align*}
		Hence,
		\begin{align*}
		\frac{\mathrm{d}}{\mathrm{d} t}\left\langle \tilde{\psi},(\mathcal{N}+1)^{j}\tilde{\psi}\right\rangle \leq C & \int\mathrm{d} x\mathrm{d} y\,|V(x-y)||\varphi_{t}(x)||\varphi_{t}(y)|\|(\mathcal{N}+3)^{j/2-1}a_{x}a_{y}\tilde{\psi}\|\notag\\
		& \qquad \times\|(\mathcal{N}+3)^{1-j/2}((\mathcal{N}+1)^{j}-(\mathcal{N}+3)^{j})\tilde{\psi}\|\notag\\
		\leq C &\|(\mathcal{N}+3)^{1-j/2}((\mathcal{N}+1)^{j}-(\mathcal{N}+3)^{j})\tilde{\psi}\|\left(\int\mathrm{d} x\mathrm{d} y\,|V(x-y)|^2|\varphi_{t}(x)|^{2}|\varphi_{t}(y)|^{2}\right)^{1/2}\notag\\
		& \times\left(\int\mathrm{d} x\mathrm{d} y\|(\mathcal{N}+3)^{j/2-1}a_{x}a_{y}\tilde{\psi}\|^{2}\right)^{1/2}.
		\end{align*}
		As in the proof of Lemma \ref{lem:decaying}, we decompose $V$ into $V_2 + V_{\infty}$ and estimate them separately. We then have
		\begin{align*}
		&\int\mathrm{d} x\mathrm{d} y\,|V(x-y)|^2|\varphi_{t}(x)|^{2}|\varphi_{t}(y)|^{2}\\
		&\qquad\leq {2} \int\mathrm{d} x\mathrm{d} y\,|V_2(x-y)|^2|\varphi_{t}(x)|^{2}|\varphi_{t}(y)|^{2} 
		+ {2}\int\mathrm{d} x\mathrm{d} y\,|V_\infty(x-y)|^2|\varphi_{t}(x)|^{2}|\varphi_{t}(y)|^{2},
		\end{align*}
		with the bounds
		\begin{equation}
		\begin{aligned}
		\int\mathrm{d} x\mathrm{d} y\,|V_2(x-y)|^2|\varphi_{t}(x)|^{2}|\varphi_{t}(y)|^{2} &\leq C \||V_2|^2\|_{L^1}\||\varphi_t|^2\|_{L^2}^2\\
		&\leq C\|V_2\|_{L^2}^2\|\varphi_t\|_{L^4}^4 \leq C\|V_2\|_{L^2}^2\|\varphi_t\|_{L^2}\|\varphi_t\|_{H^1}^{3},
		\label{eq:bound_of_HLS_type}
		\end{aligned}
		\end{equation}
		and
		\begin{equation*}
		\int\mathrm{d} x\mathrm{d} y\,|V_\infty(x-y)|^2|\varphi_{t}(x)|^{2}|\varphi_{t}(y)|^{2} \leq
		\|V_\infty\|^2_{L^\infty}\|\varphi_t\|_{L^2}^4.
		\end{equation*}
		We also have that
		\begin{align*}
		\int\mathrm{d} x\mathrm{d} y\|(\mathcal{N}+3)^{j/2-1}a_{x}a_{y}\tilde{\psi}\|^{2} & =\int\mathrm{d} x\mathrm{d} y\|a_{x}a_{y}(\mathcal{N}+1)^{j/2-1}\tilde{\psi}\|^{2}\\
		& \leq\|(\mathcal{N}+1)^{j/2}\tilde{\psi}\|^{2}.
		\end{align*}
		Altogether, we have shown that
		\begin{align*}
		\frac{\mathrm{d}}{\mathrm{d} t}\left\langle \tilde{\mathcal{U}}(t;s)\psi,(\mathcal{N}+1)^{j}\tilde{\mathcal{U}}(t;s)\psi\right\rangle
		& \leq C\|(\mathcal{N}+1)^{j/2}\tilde{\mathcal{U}}(t;s)\psi\|^{2}\\
		& =C\left\langle \tilde{\mathcal{U}}(t;s)\psi,(\mathcal{N}+1)^{j}\tilde{\mathcal{U}}(t;s)\psi\right\rangle .
		\end{align*}
		Appyling Gronwall's lemma, we conclude that
		\begin{align*}
		\left\langle \tilde{\mathcal{U}}(t;s)\psi,(\mathcal{N}+1)^{j}\tilde{\mathcal{U}}(t;s)\psi\right\rangle
		\leq Ce^{K|t-s|} \left\langle \psi,(\mathcal{N}+1)^{j}\psi\right\rangle,
		\end{align*}
		which proves the desired lemma.
	\end{proof}
	
	The main difference between the unitary operators $\mathcal{U}$ and $\tilde{\mathcal{U}}$ stems from the generator $\mathcal{L}_3$. In the following lemma, we find an estimate on $\mathcal{L}_3$.
	
	\begin{lem} \label{lem:N_1_L3}
		Suppose that the assumptions in Theorem \ref{thm:Main_Theorem} hold. Then, for any $\psi\in\mathcal{F}$ and $j\in\mathbb{N}$, there exist a constant $C \equiv C(j)$ such that
		\[
		\left\Vert \left(\mathcal{N}+1\right)^{j/2}\mathcal{L}_{3}(t)\psi\right\Vert \leq\frac{C(\|\varphi_{t}\|_{L^{\infty}(\mathbb{R}^{3})}+1)}{\sqrt{N}}\left\Vert \left(\mathcal{N}+1\right)^{\left(j+3\right)/2}\psi\right\Vert .
		\]
	\end{lem}
	
	\begin{proof}
		We basically follow the proof of Lemma 5.3 in \cite{Lee2013} {and apply} the Strichartz's estimate in Lemma \ref{lem:decaying}. Let
		\[
		A_{3}(t)=\int\mathrm{d} x\mathrm{d} y\,V(x-y)\overline{\varphi_t (y)}a_{x}^{*}a_{y}a_{x}.
		\]
		Then, by definition,
		\begin{equation}
		(\mathcal{N}+1)^{j/2}\mathcal{L}_{3}(t)=\frac{ 1}{\sqrt{N}}\left((\mathcal{N}+1)^{j/2}A_{3}(t)+(\mathcal{N}+1)^{j/2}A_{3}^{*}(t)\right).\label{eq:L3}
		\end{equation}
		We estimate $(\mathcal{N}+1)^{j/2}A_{3}(t)$ and $(\mathcal{N}+1)^{j/2}A_{3}^{*}(t)$ separately. The first term, $(\mathcal{N}+1)^{j/2}A_{3}(t)$, satisfies for any $\xi\in\mathcal{F}$ that
		\begin{align}
		& \left|\langle\xi,(\mathcal{N}+1)^{j/2}A_{3}(t)\psi\rangle\right|\notag\\
		& \quad=\left|\int\mathrm{d} x\mathrm{d} y\,V(x-y)\overline{\varphi_{t}(y)}\langle\xi,(\mathcal{N}+1)^{j/2}a_{x}^{*}a_{y}a_{x}\psi\rangle\right|\notag\\
		& \quad=\left|\int\mathrm{d} x\mathrm{d} y\,V(x-y)\overline{\varphi_{t}(y)}\langle(\mathcal{N}+1)^{-1/2}\xi,(\mathcal{N}+1)^{(j+1)/2}a_{x}^{*}a_{y}a_{x}\psi\rangle\right|\notag\\
		& \quad\leq\left(\int\mathrm{d} x\mathrm{d} y\,|V(x-y)|^2|\varphi_{t}(y)|^2\|a_{x}(\mathcal{N}+1)^{-1/2}\xi\|^{2}\right)^{1/2} \left(\int\mathrm{d} x\mathrm{d} y\|a_{y}a_{x}\mathcal{N}^{(j+1)/2}\psi\|^{2}\right)^{1/2}\notag.
		\end{align}
		Then, following the argument used in the proof of Lemma \ref{lem:NjU}, we find that
		\begin{align*}
		\left(\int\mathrm{d} x\mathrm{d} y\,|V(x-y)|^2|\varphi_{t}(y)|^2\|a_{x}(\mathcal{N}+1)^{-1/2}\xi\|^{2}\right)^{1/2} & \leq  \sup_{x}\|V(x-\cdot)\varphi_{t}\|
		\|\xi\| \\
		& \leq C(\|\varphi_{t}\|_{L^{\infty}(\mathbb{R}^{3})}+1)\|\xi\|
		\end{align*}
		Hence,
		\[
		\left|\langle\xi,(\mathcal{N}+1)^{j/2}A_{3}(t)\psi\rangle\right|\leq C(\|\varphi_{t}\|_{L^{\infty}(\mathbb{R}^{3})}+1)\|\xi\|\|\mathcal{N}^{(j+3)/2}\psi\|.
		\]
		Since $\xi$ was arbitrary, we obtain that
		\begin{equation}
		\|(\mathcal{N}+1)^{j/2}A_{3}(t)\psi\|\leq C(\|\varphi_{t}\|_{L^{\infty}(\mathbb{R}^{3})}+1)\|\mathcal{N}^{(j+3)/2}\psi\|.\label{eq:A3}
		\end{equation}
		Similarly, we can find that
		\begin{equation}
		\|(\mathcal{N}+1)^{j/2}A_{3}^{*}(t)\psi\|\leq C(\|\varphi_{t}\|_{L^{\infty}(\mathbb{R}^{3})}+1)\|(\mathcal{N}+2)^{(j+3)/2}\psi\|.\label{eq:conjA3}
		\end{equation}
		Hence, from \eqref{eq:L3}, \eqref{eq:A3}, and \eqref{eq:conjA3} we get
		\[
		\|(\mathcal{N}+1)^{j/2}\mathcal{L}_{3}(t)\psi\|\leq\frac{C(\|\varphi_{t}\|_{L^{\infty}(\mathbb{R}^{3})}+1)}{\sqrt{N}}\|(\mathcal{N}+1)^{(j+3)/2}\psi\|,
		\]
		which was to be proved.
	\end{proof}
	
	Finally, we prove the following lemma on the difference between $\mathcal{U}$ and $\tilde{\mathcal{U}}$.
	
	\begin{lem}\label{lem:NjUphiUtildeUphitildeU}
		Suppose that the assumptions in Theorem \ref{thm:Main_Theorem} hold. Then, for all $j\in\mathbb{N}$, there exist constants $C\equiv C(j)$ and $K \equiv K(t)$
		such that, for any $f\in L^{2}\left(\mathbb{R}^{3}\right)$,
		\[
		\left\Vert \left(\mathcal{N}+1\right)^{j/2}\left(\mathcal{U}^{*}(t)\phi(f)\mathcal{U}(t)-\tilde{\mathcal{U}}^{*}(t)\phi{(f)}\tilde{\mathcal{U}}(t)\right)\Omega\right\Vert \leq \frac{C \| f \| e^{Kt{^{3/2}}}}{\sqrt{N}}.
		\]
	\end{lem}
	\begin{proof}
		We follow the proof of Lemma 5.4 in \cite{Lee2013}. Let
		\[
		\mathcal{R}_{1}(f):=\left(\mathcal{U}^{*}(t)-\tilde{\mathcal{U}}^{*}(t)\right)\phi(f)\tilde{\mathcal{U}}(t)
		\]
		and
		\[
		\mathcal{R}_{2}(f):={\mathcal{U}^{*}}(t)\phi(f)\left(\mathcal{U}(t)-\tilde{\mathcal{U}}(t)\right)
		\]
		so that
		\begin{equation} \label{eq:R_1}
		\mathcal{U}^{*}(t)\phi(f)\mathcal{U}(t)-\tilde{\mathcal{U}}^{*}(t)\phi(f)\tilde{\mathcal{U}}(t)=\mathcal{R}_{1}(f)+\mathcal{R}_{2}(f).
		\end{equation}
		We begin by estimating the first term in the right-hand side of \eqref{eq:R_1}.  From Lemma \ref{lem:NjU},
		\begin{align*}
		\left\Vert (\mathcal{N}+1)^{j/2}\mathcal{R}_{1}(f)\Omega\right\Vert  & =\left\Vert \int_{0}^{t}\mathrm{d} s(\mathcal{N}+1)^{j/2}\mathcal{U}^{*}(s;0)\mathcal{L}_{3}(s)\tilde{\mathcal{U}}^{*}(t;s)\phi(f)\tilde{\mathcal{U}}(t)\Omega\right\Vert \\
		& \leq\int_{0}^{t}\mathrm{d} s\left\Vert (\mathcal{N}+1)^{j/2}\mathcal{U}^{*}(s;0)\mathcal{L}_{3}(s)\tilde{\mathcal{U}}^{*}(t;s)\phi(f)\tilde{\mathcal{U}}(t)\Omega\right\Vert \\
		& \leq Ce^{Kt{^{3/2}}}\int_{0}^{t}\mathrm{d} s\left\Vert (\mathcal{N}+1)^{j+1}\mathcal{L}_{3}(s)\tilde{\mathcal{U}}^{*}(t;s)\phi(f)\tilde{\mathcal{U}}(t)\Omega\right\Vert.
		\end{align*}
		From Lemma \ref{lem:N_1_L3} and the inequality \eqref{eq:VphiLt2Lxinfty},
		\begin{align*}
		\left\Vert (\mathcal{N}+1)^{j/2}\mathcal{R}_{1}(f)\Omega\right\Vert & \leq\frac{Ce^{Kt{^{3/2}}}}{\sqrt{N}}\int_{0}^{t}\mathrm{d} s(\|\varphi_ s\|_{L^{\infty}(\mathbb{R}^{3})}+1)\left\Vert (\mathcal{N}+1)^{j+(5/2)}\tilde{\mathcal{U}}^{*}(t;s)\phi(f)\tilde{\mathcal{U}}(t)\Omega\right\Vert \\
		& \leq\frac{Ce^{Kt{^{3/2}}}}{\sqrt{N}}\int_{0}^{t}\mathrm{d} s(\|\varphi_ s\|_{L^{\infty}(\mathbb{R}^{3})}+1)\left\Vert (\mathcal{N}+1)^{j+(5/2)}\phi(f)\tilde{\mathcal{U}}(t)\Omega\right\Vert \\
		& \leq\frac{Ce^{Kt{^{3/2}}}}{\sqrt{N}}\left(\int_{0}^{t}\mathrm{d} s\left\Vert (\mathcal{N}+1)^{j+(5/2)}\phi(f)\tilde{\mathcal{U}}(t)\Omega\right\Vert ^{2}\right)^{1/2},
		\end{align*}
		and, since the integrand in the right-hand side does not depend on $s$, we get
		\begin{align*}
		\left\Vert (\mathcal{N}+1)^{j/2}\mathcal{R}_{1}(f)\Omega\right\Vert \leq\frac{Ce^{Kt{^{3/2}}}}{\sqrt{N}}\left\Vert (\mathcal{N}+1)^{j+5/2}\phi(f)\tilde{\mathcal{U}}(t)\Omega\right\Vert.
		\end{align*}
		Thus, from Lemma \ref{lem:tildeNj}, we obtain for $\mathcal{R}_{1}(f)$ that
		\begin{align*}
		& \left\Vert (\mathcal{N}+1)^{j/2}\mathcal{R}_{1}(f)\Omega\right\Vert \\
		& \quad\leq\frac{Ce^{Kt{^{3/2}}}}{\sqrt{N}}\left(\|a(f)(\mathcal{N}+1)^{j+(5/2)}\tilde{\mathcal{U}}(t)\Omega\|+\|a^{*}(f)(\mathcal{N}+1)^{j+(5/2)}\tilde{\mathcal{U}}(t)\Omega\|\right)\\
		& \quad\leq\frac{C\|f\|e^{Kt{^{3/2}}}}{\sqrt{N}}\|(\mathcal{N}+1)^{j+{3}}\tilde{\mathcal{U}}(t)\Omega\| \leq\frac{C\|f\|e^{Kt{^{3/2}}}}{\sqrt{N}}\|(\mathcal{N}+1)^{j+{3}}\Omega\|\leq\frac{C\|f\|e^{Kt{^{3/2}}}}{\sqrt{N}}.
		\end{align*}
		The study of $\mathcal{R}_{2}(f)$ is similar and we can again obtain that
		\[
		\|(\mathcal{N}+1)^{j/2}\mathcal{R}_{2}(f)\Omega\|\leq\frac{C\|f\|e^{Kt{^{3/2}}}}{\sqrt{N}}.
		\]
		This completes the proof of the desired lemma.
	\end{proof}
	
	\subsection{Proof of Propositions \ref{prop:Et1} and \ref{prop:Et2}} \label{sec:Pf-of-Props}
	
	In this section, we prove Propositions \ref{prop:Et1} and \ref{prop:Et2} by applying the lemmas proved in Subsection \ref{subsec:comparison}.
	
	\begin{proof}[Proof of Proposition \ref{prop:Et1}]
		Recall that
		\[
		E_{t}^{1}(J)=\frac{d_{N}}{N}\left\langle W^{*}(\sqrt{N}\varphi)\frac{(a^{*}(\varphi))^{N}}{\sqrt{N!}}\Omega,\mathcal{U}^{*}(t)d\Gamma(J)\mathcal{U}(t)\Omega\right\rangle
		\]
		We begin by
		\begin{align}
		\left|E_{t}^{1}(J)\right|	&=
		\left|\frac{d_{N}}{N}\left\langle W^{*}(\sqrt{N}\varphi)\frac{(a^{*}(\varphi))^{N}}{\sqrt{N!}}\Omega,\mathcal{U}^{*}(t)d\Gamma(J)\mathcal{U}(t)\Omega\right\rangle \right| \label{eq:E_t^1 1}\\
		&\leq\frac{d_{N}}{N}\left\Vert (\mathcal{N}+1)^{-\frac{1}{2}}W^{*}(\sqrt{N}\varphi)\frac{(a^{*}(\varphi))^{N}}{\sqrt{N!}}\Omega\right\Vert \left\Vert (\mathcal{N}+1)^{\frac{1}{2}}\mathcal{U}^{*}(t)d\Gamma(J)\mathcal{U}(t)\Omega\right\Vert \notag
		\end{align}
		From Lemma \ref{lem:coherent_all},
		\begin{equation} \label{eq:E_t^1 2}
		\left\Vert (\mathcal{N}+1)^{-\frac{1}{2}}W^{*}(\sqrt{N}\varphi)\frac{(a^{*}(\varphi))^{N}}{\sqrt{N!}}\Omega\right\Vert 	\leq\frac{C}{d_{N}}.
		\end{equation}
		By successively applying Lemma \ref{lem:NjU} (and also using the inequality \eqref{eq:J-bd}), we also get
		\begin{align}
		\left\Vert (\mathcal{N}+1)^{\frac{1}{2}}\mathcal{U}^{*}(t)d\Gamma(J)\mathcal{U}(t)\Omega\right\Vert 	&\leq Ce^{Kt{^{3/2}}}\left\Vert (\mathcal{N}+1)^{2}d\Gamma(J)\mathcal{U}(t)\Omega\right\Vert \leq C\left\Vert J\right\Vert e^{Kt{^{3/2}}}\left\Vert (\mathcal{N}+1)^{3}\mathcal{U}(t)\Omega\right\Vert \notag \\
		&\leq C\left\Vert J\right\Vert e^{Kt{^{3/2}}}\left\Vert (\mathcal{N}+1)^{7}\Omega\right\Vert. \label{eq:E_t^1 3}
		\end{align}
		Thus, from \eqref{eq:E_t^1 1}, \eqref{eq:E_t^1 2}, and \eqref{eq:E_t^1 3},
		\[
		\left|E_{t}^{1}(J)\right|\leq\frac{C\|J\|e^{Kt{^{3/2}}}}{N},
		\]
		which proves the desired result.
	\end{proof}
	
	For the proof of Proposition \ref{prop:Et2}, we take almost verbatim copy of the proof of Lemma 4.2 in \cite{Lee2013}. To make the paper self-contained, we write it in detail below.
	
	\begin{proof}[Proof of Proposition \ref{prop:Et2}]
		
		Recall the definitions of $\mathcal{R}_1$ and $\mathcal{R}_2$ in the proof of Lemma \ref{lem:NjUphiUtildeUphitildeU}. Let $\mathcal{R} = \mathcal{R}_1 + \mathcal{R}_2$ so that
		\[
		\mathcal{R}(f)=\mathcal{U}^{*}(t)\phi(f)\mathcal{U}(t)-\tilde{\mathcal{U}}^{*}(t)\phi(f)\tilde{\mathcal{U}}(t).
		\]
		From the parity conservation \eqref{eq:Parity_Consevation},
		\[
		P_{2k}\tilde{\mathcal{U}}^{*}(t)\phi(J\varphi_{t})\tilde{\mathcal{U}}(t)\Omega=0
		\]
		for all $k=0,1,{{\dots}}$. (See Lemma 8.2 in \cite{Lee2013} for more detail.) Thus,
		\begin{align}
		\left|E_{t}^{2}(J)\right| & =\frac{d_{N}}{\sqrt{N}}\left\langle \frac{(a^{*}(\varphi))^{N}}{\sqrt{N!}}\Omega,W^{*}(\sqrt{N}\varphi)\mathcal{\tilde{U}}^{*}(t)\phi(J\varphi_{t})\mathcal{\tilde{U}}(t)\Omega\right\rangle \notag \\
		& \qquad+\frac{d_{N}}{\sqrt{N}}\left\langle \frac{(a^{*}(\varphi))^{N}}{\sqrt{N!}}\Omega,W^{*}(\sqrt{N}\varphi)\mathcal{R}(J\varphi_{t})\Omega\right\rangle \notag \\
		& \leq\frac{d_{N}}{\sqrt{N}}\left\Vert \sum_{k=1}^{\infty}(\mathcal{N}+1)^{-\frac{5}{2}}P_{2k-1}W^{*}(\sqrt{N}\varphi)\frac{(a^{*}(\varphi))^{N}}{\sqrt{N!}}\Omega\right\Vert \left\Vert (\mathcal{N}+1)^{\frac{5}{2}}\mathcal{\tilde{U}}^{*}(t)\phi(J\varphi_{t})\tilde{\mathcal{U}}(t)\Omega\right\Vert \notag \\
		& \qquad+\frac{d_{N}}{\sqrt{N}}\left\Vert (\mathcal{N}+1)^{-\frac{1}{2}}W^{*}(\sqrt{N}\varphi)\frac{(a^{*}(\varphi))^{N}}{\sqrt{N!}}\Omega\right\Vert \left\Vert (\mathcal{N}+1)^{\frac{1}{2}}\mathcal{R}(J\varphi_{t})\Omega\right\Vert \label{eq:e_t^2 1}
		\end{align}
		Let $K=\frac{1}{2}N^{1/3}$ so that Lemmas \ref{lem:coherent_all} and \ref{lem:coherent_even_odd} show that
		\begin{align*}
		& \left\Vert \sum_{k=1}^{\infty}(\mathcal{N}+1)^{-\frac{5}{2}}P_{2k-1}W^{*}(\sqrt{N}\varphi)\frac{(a^{*}(\varphi))^{N}}{\sqrt{N!}}\Omega\right\Vert ^{2}\\
		& \qquad\leq\sum_{k=1}^{K}\left\Vert (\mathcal{N}+1)^{-\frac{5}{2}}P_{2k-1}W^{*}(\sqrt{N}\varphi)\frac{(a^{*}(\varphi))^{N}}{\sqrt{N!}}\Omega\right\Vert ^{2}\\
		& \qquad\qquad+\frac{1}{K^{4}}\sum_{k=K}^{\infty}\left\Vert(\mathcal{N}+1)^{-1/2} P_{2k-1}W^{*}(\sqrt{N}\varphi)\frac{(a^{*}(\varphi))^{N}}{\sqrt{N!}}\Omega\right\Vert ^{2}\\
		& \qquad\leq\left(\sum_{k=1}^{K}\frac{C}{k^{2}d_{N}^{2}N}\right)+\frac{C}{N^{4/3}}\left\Vert(\mathcal{N}+1)^{-1/2} W^{*}(\sqrt{N}\varphi)\frac{(a^{*}(\varphi))^{N}}{\sqrt{N!}}\Omega\right\Vert \leq\frac{C}{d_{N}^{2}N}.
		\end{align*}
		Using Lemma \ref{lem:tildeNj},
		\begin{alignat*}{1}
		& \left\Vert (\mathcal{N}+1)^{\frac{5}{2}}\mathcal{\tilde{U}}^{*}(t)\phi(J\varphi_{t})\tilde{\mathcal{U}}(t)\Omega\right\Vert \leq Ce^{Kt{^{3/2}}}\left\Vert (\mathcal{N}+1)^{\frac{5}{2}}\phi(J\varphi_{t})\tilde{\mathcal{U}}(t)\Omega\right\Vert \\
		& \quad\leq C\|J\varphi_{t}\|e^{Kt{^{3/2}}}\left\Vert (\mathcal{N}+1)^{3}\mathcal{\tilde{U}}(t)\Omega\right\Vert \leq C\|J\|e^{Kt{^{3/2}}}\left\Vert (\mathcal{N}+1)^{3}\Omega\right\Vert \leq C\|J\|e^{Kt{^{3/2}}}.
		\end{alignat*}
		For the second term in the right-hand side {of \eqref{eq:e_t^2 1}}, we use Lemmas \ref{lem:coherent_all} and \ref{lem:NjUphiUtildeUphitildeU}, where we put $J\varphi_{t}$ in the place of $f$ for the latter. Altogether, we have
		\begin{equation*}
		\left\Vert (\mathcal{N}+1)^{j/2}\mathcal{R}(f)\Omega\right\Vert \leq \frac{C\|f\|e^{Kt{^{3/2}}}}{N},
		\end{equation*}
		which is the desired conclusion.
	\end{proof}

	\section*{Acknowledgements}
	
	We would like to thank Benjamin Schlein for helpful discussion. We are also grateful to the anonymous referee for carefully reading our manuscript and providing helpful comments. Li Chen is supported by Deutsche Forschungsgemeinschaft (DFG) Grant CH 955/4-1. Ji Oon Lee and Jinyeop Lee are supported in part by the Basic Science Research Program through the National Research Foundation of Korea grant 2011-0013474.

	\bibliographystyle{myplain}
	\bibliography{refs}
	
\end{document}